\documentclass[11pt]{article} 

\usepackage[margin=1in]{geometry} 
\usepackage[utf8]{inputenc}
\usepackage[T1]{fontenc}
\usepackage{mathpazo} 
\usepackage{courier} 
\usepackage{microtype}
\ifdefined\pdfinfoomitdate
\fi
\ifdefined\pdfsuppressptexinfo
\fi
\ifdefined\pdftrailerid
  \pdftrailerid{}
\fi

\usepackage{comment,xspace}
\usepackage[dvipsnames]{xcolor} 
\usepackage{booktabs}

\usepackage{enumitem} 
\definecolor{ptblue}{RGB}{15,76,129} 
\definecolor{ptemerald}{HTML}{009473} 
\definecolor{ptgray}{HTML}{939597} 
\definecolor{cobalt}{rgb}{0.0, 0.28, 0.67}
\usepackage{amsmath,amsfonts,amssymb,amsthm,mathtools}
\let\OLDforall\forall
\renewcommand{\forall}{\;\OLDforall\:}
\let\OLDexists\exists
\renewcommand{\exists}{\;\OLDexists\,}

\DeclareMathOperator*{\argmax}{arg\,max}

\usepackage{array}

\usepackage[square,numbers,sort]{natbib}
\usepackage{hyperref}
\usepackage{aliascnt}
\hypersetup{
linktocpage,
colorlinks=true,
citecolor=cobalt, 
urlcolor=ptblue, 
linkcolor=Plum, 
pdfcreator={},
pdfproducer={},
}
\usepackage{cleveref} 
\usepackage{crossreftools} 
\theoremstyle{plain}
\newtheorem{theorem}{Theorem}[section]

\newaliascnt{conjecture}{theorem}
\newtheorem{conjecture}[conjecture]{Conjecture}
\aliascntresetthe{conjecture}

\newaliascnt{openproblem}{theorem}

\aliascntresetthe{openproblem}

\newaliascnt{corollary}{theorem}
\newtheorem{corollary}[corollary]{Corollary}
\aliascntresetthe{corollary}

\newaliascnt{lemma}{theorem}

\aliascntresetthe{lemma}

\newaliascnt{proposition}{theorem}
\newtheorem{proposition}[proposition]{Proposition}
\aliascntresetthe{proposition}

\theoremstyle{definition}
\newaliascnt{definition}{theorem}
\newtheorem{definition}[definition]{Definition}
\aliascntresetthe{definition}

\newaliascnt{example}{theorem}

\aliascntresetthe{example}

\newaliascnt{observation}{theorem}
\newtheorem{observation}[observation]{Observation}
\aliascntresetthe{observation}

\theoremstyle{remark}
\newtheorem*{remark}{\upshape\bfseries Remark}

\crefname{theorem}{Theorem}{Theorems}
\crefname{conjecture}{Conjecture}{Conjectures}
\crefname{openproblem}{Open Problem}{Open Problems}
\crefname{corollary}{Corollary}{Corollaries}
\crefname{lemma}{Lemma}{Lemmas}
\crefname{proposition}{Proposition}{Propositions}
\crefname{definition}{Definition}{Definitions}
\crefname{example}{Example}{Examples}
\crefname{observation}{Observation}{Observations}

\makeatletter 
\newcommand{\EF}[1]{\if\relax\detokenize\expandafter{\@firstofone#1{}}\relax \text{EF}\xspace\else \text{EF#1}\fi}
\makeatother
\newcommand{\EFone}{\EF{1}\xspace}
\newcommand{\EFX}{\EF{X}\xspace}

\newcommand{\PO}{\text{PO}\xspace}
\newcommand{\wPO}{\text{weak-PO}\xspace}
\newcommand{\fPO}{\text{fPO}\xspace}
\newcommand{\EFonePO}{\textnormal{EF1+PO}\xspace}
\newcommand{\EFonefPO}{\textnormal{EF1+fPO}\xspace}

\newcommand{\NSW}{\text{NSW}\xspace}

\newcommand{\GS}{\text{Gross Substitutes}\xspace}
\newcommand{\CWMR}{\text{Common Weight Matroid Rank}\xspace}
\DeclareRobustCommand{\CWMR}{%
  \ifmmode \mathrm{CWMR}\else Common Weighted Matroid Rank\fi\xspace}
\newcommand{\WMR}{\text{Weighted Matroid Rank}\xspace}

\newcommand{\RR}{\mathbb{R}}
\newcommand{\cC}{\mathcal{C}}

\newcommand{\M}{\mathcal M}
\newcommand{\I}{\mathcal I}

\title{When One Good Is Not Enough: EF1 and Pareto Optimality Are Not Compatible for Submodular Valuations}

\author{%
  Simon Mackenzie%
  \thanks{Email: \href{mailto:simon.william.mackenzie@gmail.com}
{\nolinkurl{simon.william.mackenzie@gmail.com}}}
\and 
  Mashbat Suzuki%
  \thanks{Email: \href{mailto:mashbat.suzuki@unsw.edu.au}
  {\nolinkurl{mashbat.suzuki@unsw.edu.au}}}
}

\date{}

\begin{document}
\hypersetup{
  linktocpage,
  colorlinks=true,
  citecolor=cobalt,
  urlcolor=ptblue,
  linkcolor=Plum,
  pdftitle={When One Good Is Not Enough: EF1 and Pareto Optimality Are Not Compatible for Submodular Valuations},
  pdfauthor={Mashbat Suzuki and Simon Mackenzie},
  pdfcreator={},
  pdfproducer={}
}
\begin{titlepage}
\maketitle
\thispagestyle{empty}

\begin{abstract}

One of the central questions in discrete fair division is whether fairness and efficiency can be achieved simultaneously. For indivisible goods, a canonical relaxation of envy-freeness is envy-freeness up to one good (\EFone), while the standard efficiency benchmark is Pareto optimality (\PO). In their seminal work, \citet{CKM16} showed that, for additive valuations, \EFone and \PO are always compatible, and asked whether this compatibility extends to submodular valuations. This question has since become an important open problem in the study of fair division \citep{caragiannis2019unreasonable,benabbou2021finding,BS2026,KKM23}.

In this paper, we settle the question in the negative. We construct an instance with two agents and eight goods, where both agents have monotone submodular valuations, such that no \EFone allocation is even weakly Pareto optimal. Thus, the celebrated compatibility between \EFone and \PO for additive valuations breaks down already for two agents under submodular valuations. 

We then map the boundary of this impossibility. On the negative side, we show that even for weighted matroid rank valuations, a highly structured subclass of submodular valuations, \EFone and fractional Pareto optimality (\fPO) are incompatible. This rules out, in general, broad classes of weighted-welfare and Fisher-market-based approaches. On the positive side, we identify a common-envelope condition that restores compatibility. Under this condition, \EFonePO{} allocations exist for any number of agents. This yields new positive results showing that common-weight matroid-rank valuations always admit \EFonePO{} allocations.

Finally, we quantify the efficiency loss that is unavoidable when insisting on \EFone. Our submodular counterexample implies that there is a constant $\alpha<1$ such that no \EFone allocation is $\alpha$-\PO. For the broader class of subadditive valuations, we prove a tight two-agent bound, for any $\varepsilon>0$, there exists an instance in which no \EFone allocation is $\left(\frac{1}{\sqrt{2}}+\varepsilon\right)$-\PO.

\end{abstract}

\vspace{0.75em}

\noindent\textbf{AI disclosure.}
The authors used large language models as research assistants primarily to accelerate computational exploration, including generating and modifying search scripts and iterating quickly over experimental directions. All reported instances, computational outputs, and final wording were checked by the authors and are the authors' responsibility.

\end{titlepage}
\hypersetup{pageanchor=true}
\setcounter{page}{1}

\section{Introduction}

Fair division of resources is governed by a basic tension. On the one hand, a
fair allocation should not leave any agent feeling that another agent has
received a substantially better bundle. On the other hand, an allocation should
not waste resources: if every agent can be made weakly better off and at least one agent strictly better off, then the original allocation should not be regarded as satisfactory. Much of the modern theory of
fair division asks whether these two desiderata---fairness and efficiency---can
be reconciled.

For divisible goods, this tension is often resolved through concavity and market
equilibrium arguments \citep{Var1974,Thoms2011}. For indivisible goods, however,
envy-free allocations may fail to exist even before efficiency is imposed, for
the simple reason that a single valuable good cannot be split. Thus, one must
relax envy-freeness in a way that accounts for indivisibilities. The most
prominent such relaxation is envy-freeness up to one good (\EFone), which states that whenever
agent \(i\) envies agent \(j\), the envy can be eliminated by removing a single
good from \(j\)'s bundle. \EFone{} has become a central fairness benchmark because it captures the idea that any remaining envy can be attributed to the indivisibility of a single item, while still being guaranteed to exist for arbitrary monotone valuations \citep{LMM04}.

On the efficiency side, the canonical requirement is Pareto optimality (\PO).
An allocation is Pareto optimal if no other allocation makes every agent weakly
better off and at least one agent strictly better off. Pareto optimality is a
minimal efficiency condition, it does not impose a social objective, cardinal
comparisons across agents, or any external notion of welfare. Thus, the
compatibility of \EFone{} and \PO{} asks for a very basic guarantee: can we
always allocate indivisible goods so that the outcome is fair up to one item and
contains no obvious waste?

For additive valuations, the answer is yes, \citet{CKM16} showed that
\EFone{} and \PO{} allocations always exist. They also showed that this
compatibility breaks down for subadditive valuations, but left open the
intermediate case of submodular valuations. This case is especially natural as
submodularity captures diminishing marginal returns and is therefore a discrete
analogue of concavity, the condition that underlies many positive results for
divisible goods. The existence of \EFone{} and \PO{} allocations under
submodular valuations has consequently remained an important open problem and
has been repeatedly highlighted in the literature
\citep{caragiannis2019unreasonable,benabbou2021finding,BS2026,KKM23}.

\medskip

We answer this question in the negative. 

\begin{theorem}
There exists a two agent and eight good instance with monotone submodular valuations such that no allocation is both \EFone and \PO. 
\end{theorem}

In fact, we prove a stronger form of incompatibility. Our instance shows that,
for submodular valuations, there is a constant \(\alpha<1\) such that no
allocation is both \EFone{} and \(\alpha\)-\PO{}.\footnote{For a parameter
\(\alpha \in [0,1]\), an allocation \(A\) is said to be \(\alpha\)-\PO{} if
there is no allocation \(B\) such that, for every agent \(i\), \(\alpha\) times
agent \(i\)'s value for \(B_i\) is strictly larger than agent \(i\)'s value for
\(A_i\). When \(\alpha=1\), this coincides with weak Pareto optimality.}
Thus, the incompatibility persists even for an approximate efficiency notion,
every \EFone{} allocation can be simultaneously improved for all agents by a
multiplicative factor bounded away from one.

Motivated by this negative result, we next seek to map the boundary of the
incompatibility. A natural direction is to ask whether the obstruction persists
for highly structured subclasses of submodular valuations. One such class is common-weight matroid-rank valuations. In this class, each agent $i$ may have a different matroid \(\M_i=(M,\I_i)\) over the ground set of goods \(M\), but they agree on common weight function $w:M \to \RR_{\ge 0}$; the value of a set
is for agent $i$ given by its weighted matroid rank $v_i(S):=\max\left\{w(I): I\subseteq S,\ I\in\I_i\right\}$. Since these valuations are far more structured than general
submodular valuations, one might expect them to be amenable to the standard
techniques that succeed in the additive setting.

We show, however, that these standard techniques cannot be extended in a
black-box manner. Fisher-market methods \citep{barman2018finding} and
weighted-welfare-based methods \citep{mahara2025chores,BarmanIntros} run into a
fundamental obstruction, as we show that \EFone and \fPO are not compatible.
Specifically, we construct an instance with common-weight matroid-rank
valuations such that, for every choice of strictly positive welfare weights,
every allocation that maximizes the corresponding weighted welfare fails to
satisfy \EFone{}. Thus, the obstruction cannot be avoided by choosing different
welfare weights or by breaking ties carefully. Moreover, this failure is robust,
it persists under arbitrarily small perturbations of the valuations.

This obstruction is methodological rather than existential. Although
weighted-welfare maximization does not suffice by itself, we develop a different
approach based on a structural property that we call the \emph{common-envelope
condition}. This condition restores the compatibility of fairness and
efficiency for common-weight matroid-rank valuations, yielding the following
positive result.

\begin{theorem}
For any number of agents with common-weight matroid-rank valuations, there
exists an allocation that is both \EFone{} and \PO{}.
\end{theorem}

As noted earlier, our submodular counterexample proves more than the
incompatibility of \EFone{} and \PO{}. It shows that, for some constant
\(\alpha<1\), no allocation is both \EFone{} and \(\alpha\)-\PO{}. Since every
submodular valuation is subadditive, this immediately yields the same
incompatibility for subadditive valuations.

This strengthens the subadditive impossibility of \citet{CKM16}. They showed
that \EFone{} and \PO{} need not be compatible for subadditive valuations, but
their instance still admits an allocation that is both \EFone{} and weak-\PO. In contrast, our instance is the first known subadditive example
in which \EFone{} is incompatible with \(\alpha\)-\PO{} for some
\(\alpha<1\).

We then ask how far this incompatibility can be pushed. By the result of
\citet{BS2026}, every subadditive instance admits an allocation that is both
\EFone{} and \(\frac{1}{2}\)-\PO{}. Thus, the best possible approximation
threshold lies somewhere between \(\frac{1}{2}\) and \(1\). We give the first
nontrivial upper bound on this threshold.

\begin{theorem}
For every \(n\geq 2\) and every \(\varepsilon>0\), there exists an \(n\)-agent
instance with monotone subadditive valuations such that no allocation is both
\EFone{} and \(\left(\frac{1}{\sqrt{2}}+\varepsilon\right)\)-\PO{}.
\end{theorem}

There is a general route from Nash social welfare to approximate Pareto
optimality. If an allocation achieves a \(\beta\)-approximation to the optimal
\NSW{}, then it is \(\beta\)-\PO{}. Indeed, if another allocation improved every
agent by a multiplicative factor larger than \(1/\beta\), then it would also
improve the Nash social welfare by a factor larger than \(1/\beta\),
contradicting \(\beta\)-optimality. Thus, any lower bound on the best
\NSW{} approximation achievable by an \EFone{} allocation immediately yields a
lower bound on the best approximate Pareto optimality achievable together with
\EFone{}.

For two agents, this connection exactly matches our upper bound.
\citet[Theorem A.2]{BS2026} show that every two-agent subadditive instance
admits an \EFone{} allocation whose \NSW{} is at least \(\frac{1}{\sqrt{2}}\) times
optimal. By the observation above, such an allocation is
\(\frac{1}{\sqrt{2}}\)-\PO{}. Our impossibility result shows that this factor is
best possible. We believe this is not a coincidence, but rather evidence of a
deeper connection between approximate efficiency and Nash social welfare in the
fair allocation of indivisible goods with subadditive valuations.

\medskip 

Due to space constraints, some proofs are deferred to the appendix.

\medskip 

\noindent\textbf{Organization.} \Cref{sec:prelim} gives definitions.  Theorem~1.1 is
restated as \Cref{thm:submodular-ce} and proved in
\Cref{sec:core-submodular}, which also contains the even-agent weak-\PO{}
extension and the relabeling examples.  \Cref{sec:subclasses} proves the
\fPO{} barriers and the common-envelope theorem; Theorem~1.2 follows from
\Cref{cor:common-weight-matroid}.  \Cref{sec:approx} proves Theorem~1.3,
restated as \Cref{thm:subadditive-alpha,thm:subadditive-alpha-alln}.

\medskip 

\noindent\textbf{Related Work.}
For divisible goods, equilibrium and Nash-welfare arguments give envy-free and
Pareto-optimal outcomes \citep{Var1974,Thoms2011}.  For indivisible goods,
\EFone{} is the standard relaxation that is studied widely in the fair division literature \citep{LMM04,Bud11}.  For additive valuations,  maximum Nash welfare solution always satisfies \EFone{} and \PO{} \citep{CKM16,caragiannis2019unreasonable}.  However, Nash welfare maximization is not computationally tractable. To remedy this, \citet{barman2018finding} showed that there exists a pseudo-polynomial time algorithm to achieve \EFone and \PO using a market-based approach. Further market-based
approaches has been subsequently explored \citep{murhekar2021fair,garg2023fewvalues, DBLP:conf/wine/Mahara25}.

Beyond additivity, \citet{caragiannis2019unreasonable} showed subadditive
incompatibility, proved that for submodular valuations Nash welfare optimal allocations satisfy  marginal-\EFone{}.  Approximate
fairness-and-Nash-welfare tradeoffs were developed for additive valuations with
charity by \citet{caragiannis2019efxcharity}, and for subadditive valuations by
\citet{chaudhury2021subadditive,BS2026}.

Gross substitutes \citep{kelso1982job,gul1999walrasian,paesleme2017gross} are
a central structured subclass of submodular valuations.  They include OXS
valuations \citep{lehmann2006combinatorial} and weighted matroid-rank
valuations \citep{paesleme2017gross,schrijver2003combinatorial}.  Fair and
efficient allocation is much better understood for unweighted matroid-rank
valuations: prior work gives near-fairness in matroidal settings, \EFone{}+\PO{}
and Lorenz-dominating allocations, and efficient Yankee-swap-type algorithms
\citep{gourves2014near,benabbou2021finding,babaioff2021dichotomous,viswanathan2023yankee,viswanathan2023general}.
Related guarantees are known for assignment valuations
\citep{KKM23}.  At the same time, gross substitutes form only one structured
part of the broader submodular landscape \citep{dobzinski2021coverage}.  We show
that arbitrary submodular valuations are too broad for \EFone{}+\PO{}, but
recover compatibility for common-weight matroid rank via the common-envelope
theorem.

Our approximate-efficiency results are related to, but distinct from, the
price-of-fairness literature, which measures loss in aggregate welfare under
fairness constraints \citep{bei2021price,barman2020price}.  We instead ask for
coordinate-wise approximate Pareto optimality.  \citet{BS2026} prove
\EFone{} and \(\frac12\)-\PO{} allocation always exists for subadditive valuations and, for two agents, an
\EFone{} allocation with \(\frac{1}{\sqrt2}\)-approximate Nash welfare exists, implying
\(\frac{1}{\sqrt2}\)-\PO{}. Our construction for two agents gives the matching upper bound.
Welfare-approximation work for subadditive valuations provides a parallel
benchmark for this Nash-welfare route \citep{barman2020pmean}.  Finally,
market, spending-restricted, and positive weighted-welfare methods often target
\fPO{}: they yield \EFone{}+\fPO{} for additive goods and have analogues for
chores and mixed manna
\citep{barman2018finding,garg2022bivaluedchores,ebadian2022chores,mahara2025chores,aziz2018goodschores,bhaskar2021approximate,liu2023mixedSurvey,garg2024mixedmanna,BarmanIntros}.
Our \fPO{} examples show that this target can be too strong even when ordinary
\EFone{} and \PO{} allocations exist.

\section{Preliminaries}\label{sec:prelim}

We study fair division of a finite set \(M\) of indivisible goods among a set
of agents \(N=\{1,\ldots,n\}\).  The preferences of agent \(i\in N\) are
represented by a valuation, or set function, \(v_i:2^M\to\RR_{\geq 0}\).  We always
normalize valuations so that \(v_i(\emptyset)=0\).  Unless stated otherwise, we
work with goods instances, where valuations are nonnegative and monotone:
\(v_i(S)\ge0\) for all \(S\subseteq M\), and \(v_i(S)\le v_i(T)\) whenever
\(S\subseteq T\subseteq M\).  We write \(v_i(g)\) as shorthand for
\(v_i(\{g\})\).  The mixed-manna example in \Cref{thm:mixed-fpo-ce} allows
signed additive valuations.  Allocation and efficiency notions are unchanged
there; the mixed-manna convention for \EFone{} is stated below.

An allocation \(X=(X_i)_{i\in N}\) is a partition of \(M\), where \(X_i\) is the
bundle assigned to agent \(i\).  We denote by \(\mathcal A\) the set of all
allocations.  A partial allocation is a tuple of pairwise disjoint bundles whose
union may be a strict subset of \(M\).

The fairness and efficiency notions used in the paper are defined next.

\begin{definition}[\EFone]
An allocation \(X=(X_i)_{i\in N}\) is envy-free up to one good (\EFone) if, for
every ordered pair of agents \(i,j\in N\), either \(X_j=\emptyset\), or there
exists a good \(g\in X_j\) such that
\[
  v_i(X_i)\ge v_i(X_j\setminus\{g\}).
\]

\end{definition}

\begin{definition}[Pareto notions]
An allocation \(Y\) Pareto-dominates an allocation \(X\) if
\(v_i(Y_i)\ge v_i(X_i)\) for every agent \(i\in N\), and the inequality is
strict for at least one agent.  An allocation is Pareto optimal (\PO) if it is
not Pareto-dominated.  It is weakly Pareto optimal (\wPO) if there is no
allocation \(Y\) such that \(v_i(Y_i)>v_i(X_i)\) for every agent \(i\in N\).
\end{definition}

For an allocation \(X\), its utility vector is
\((v_i(X_i))_{i\in N}\).  An allocation is leximin-optimal if its utility
vector, sorted in nondecreasing order, is lexicographically maximum among all
allocations.

\begin{definition}[\(\alpha\)-Pareto optimality]
For \(0<\alpha\le1\), an allocation \(X\) is \(\alpha\)-Pareto optimal if there
is no allocation \(Y\) such that
\[
  v_i(Y_i)>\frac{1}{\alpha}\,v_i(X_i)
  \qquad\text{for every agent } i.
\]
For \(\alpha=1\), this is exactly weak Pareto optimality, it rules out an
allocation that strictly improves every agent.  When \(\alpha<1\), the condition
is weaker, because it only rules out simultaneous improvements by a factor
strictly larger than \(1/\alpha\).  Some works use a stronger coordinate-wise
version, where every agent must be weakly improved by the factor \(1/\alpha\)
and at least one agent strictly improved; any allocation satisfying that
stronger version also satisfies the convention used here.
\end{definition}

\begin{definition}[Fractional Pareto optimality]
A deterministic allocation \(X\) is fractionally Pareto optimal (\fPO) if it is
not Pareto-dominated by any lottery over deterministic allocations, where each
agent evaluates a lottery by expected utility.  Equivalently, there do not
exist allocations \(Y^1,\ldots,Y^k\) and coefficients
\(\lambda_1,\ldots,\lambda_k > 0\), \(\sum_\ell\lambda_\ell=1\), such that
\[
  \sum_{\ell=1}^k \lambda_\ell v_i(Y_i^\ell)\ge v_i(X_i)
  \quad\forall i,
  \qquad
  \sum_{\ell=1}^k \lambda_\ell v_j(Y_j^\ell)> v_j(X_j)
  \quad\text{for some }j.
\]
\end{definition}

We use the following standard characterization of \fPO, which follows from
Farkas' lemma.

\begin{proposition}[fPO and weighted welfare]\label{prop:fpo-welfare}
An allocation \(X\) is \fPO if and only if there are strictly positive weights
\(\lambda_1,\ldots,\lambda_n>0\) such that
\[
  X\in\argmax_{Y\in\mathcal A}\sum_{i\in N}\lambda_i v_i(Y_i).
\]
\end{proposition}

\paragraph{Valuation classes.}
We use the following standard valuation classes.  A valuation \(v\) is additive
if \(v(S)=\sum_{g\in S}v(g)\) for every \(S\subseteq M\).  It is subadditive if
\(v(S\cup T)\le v(S)+v(T)\) for all \(S,T\subseteq M\).  It is submodular if it
has diminishing marginal returns:
\[
  v(S\cup\{g\})-v(S)\ge v(T\cup\{g\})-v(T)
  \quad\text{whenever } S\subseteq T,\ g\notin T.
\]

A valuation is gross substitutes (\GS) if the demand for goods whose prices do
not increase can be preserved when other prices rise. For any price vectors
\(p'\ge p\) and any demanded bundle \(S\in\argmax_T(v(T)-p(T))\), there is a
demanded bundle \(S'\in\argmax_T(v(T)-p'(T))\) such that
\(S\cap\{g:p'_g=p_g\}\subseteq S'\).

One important subclass of gross substitutes is weighted matroid rank (WMR).  A
WMR valuation is specified by a matroid \((M,\mathcal I)\) and nonnegative item
weights \(w\), with
\[
  v(S)=\max\{w(I): I\subseteq S,\ I\in\mathcal I\}.
\]
In a common-weight matroid-rank profile, the agents may have different matroids
but share the same item weights.
Weighted matroid-rank valuations are gross substitutes, and gross-substitutes
valuations are submodular
\citep{kelso1982job,gul1999walrasian}.

\[
\text{Additive} \subsetneq \text{Weighted Matroid Rank}
\subsetneq \text{GS} \subsetneq \text{Submodular}
\subsetneq \text{Subadditive} \, . 
\]

\section{The Core Submodular Obstruction}\label{sec:core-submodular}
We begin with the main counterexample.  The instance only has two agents and eight goods of two types $A$ and $B$. Agents are indifferent among goods of the same type, and hence the valuations can be represented by tables indexed by the numbers of $A$- and $B$-goods received.
 Fairness restricts the outcome to one of four balanced allocations, which we refer to as \emph{splits}, yet each such split is strictly dominated by an unbalanced one.

\medskip 

\noindent\textbf{Instance Description.} There are two agents and eight goods of
two types,
\[
  A=\{a_1,a_2,a_3\},\qquad
  B=\{b_1,b_2,b_3,b_4,b_5\}.
\]
Both agents are indifferent among goods of the same type, so a value of a  bundle \(S\) is only determined by the type counts
\[
  x(S)=|S\cap A|,\qquad y(S)=|S\cap B|.
\]
This allows us to specify each agents valuations via a $4\times 6$ grid.

A coarse grained description of the instance is that the valuations
rise quickly and then flatten out, no bundle is worth more than
\(3+6\varepsilon\) to either agent, and any four of the eight goods are
already worth at least \(3\). 
Once each agent holds four goods, the two
utilities are therefore confined to the narrow band
\([3,\,3+6\varepsilon]\), and everything is decided by the
\(\varepsilon\)-terms.  These are arranged to work against each other.
\EFone will hold only at four \emph{balanced} splits, in which each agent
receives exactly four goods, while each balanced split will be strictly
worse, for both agents at once, than some unbalanced one.  Fairness forces
balance; efficiency punishes it.

Fix \(0<\varepsilon\le 1/6\).  Agent 1's value table is
\[
\begin{array}{c|r@{}l r@{}l r@{}l r@{}l r@{}l r@{}l}
g_1(x,y)
  & \multicolumn{2}{c}{y=0}
  & \multicolumn{2}{c}{y=1}
  & \multicolumn{2}{c}{y=2}
  & \multicolumn{2}{c}{y=3}
  & \multicolumn{2}{c}{y=4}
  & \multicolumn{2}{c}{y=5} \\
\hline
x=0 & 0&& 1&& 2&& 3&& 3&+3\varepsilon& 3&+6\varepsilon \\
x=1 & 2&& 3&& 3&+4\varepsilon& 3&+5\varepsilon
  & 3&+6\varepsilon& 3&+6\varepsilon \\
x=2 & 2&+\varepsilon& 3&+\varepsilon& 3&+5\varepsilon
  & 3&+6\varepsilon& 3&+6\varepsilon& 3&+6\varepsilon \\
x=3 & 2&+\varepsilon& 3&+\varepsilon& 3&+5\varepsilon
  & 3&+6\varepsilon& 3&+6\varepsilon& 3&+6\varepsilon
\end{array}
\]
and agent 2's value table is
\[
\begin{array}{c|r@{}l r@{}l r@{}l r@{}l r@{}l r@{}l}
g_2(x,y)
  & \multicolumn{2}{c}{y=0}
  & \multicolumn{2}{c}{y=1}
  & \multicolumn{2}{c}{y=2}
  & \multicolumn{2}{c}{y=3}
  & \multicolumn{2}{c}{y=4}
  & \multicolumn{2}{c}{y=5} \\
\hline
x=0 & 0&& 1&& 2&& 3&& 3&+5\varepsilon& 3&+5\varepsilon \\
x=1 & 1&+2\varepsilon& 2&+\varepsilon& 3&& 3&+3\varepsilon
  & 3&+6\varepsilon& 3&+6\varepsilon \\
x=2 & 2&+3\varepsilon& 3&& 3&+3\varepsilon
  & 3&+6\varepsilon& 3&+6\varepsilon& 3&+6\varepsilon \\
x=3 & 3&+4\varepsilon& 3&+5\varepsilon& 3&+6\varepsilon
  & 3&+6\varepsilon& 3&+6\varepsilon& 3&+6\varepsilon
\end{array}
\]
Each agent $i\in \{1,2\}$ has a valuation function
\[
  v_i(S)=g_i(x(S),y(S)) \quad \text{ for all } S\subseteq A\cup B
\]
We now show that the instance described above does not admit any allocation that is both \EFone and \PO.

\begin{theorem}\label{thm:submodular-ce}
There exists an instance with two agents and eight goods, where both agents have
monotone submodular valuations, such that, for every
\(\alpha >  \frac{23}{24}\), no allocation is both \EFone{} and
\(\alpha\)-\PO{}. In particular, no allocation is both \EFone{} and \PO{}
\end{theorem}

\begin{proof}
We first verify that the valuations are monotone and submodular. Since each
valuation depends only on the counts \((x,y)\), monotonicity is equivalent to
all \(A\)- and \(B\)-marginals being nonnegative. Similarly, submodularity is
equivalent to the diminishing-returns condition that each marginal is
nonincreasing as either count increases. The marginal grids are displayed in
Appendix~\ref{app:marginals}; in each grid, all entries are nonnegative and
each row and column is nonincreasing. Hence both valuations are monotone
submodular.

We now show that no allocation can be both \EFone{} and \wPO{}. We first
identify all allocations satisfying \EFone{}. A split \((x,y)\) means that
agent~1 receives \(x\) goods of type \(A\) and \(y\) goods of type \(B\), while
agent~2 receives the remaining \((3-x,5-y)\) goods. For agent~1, deleting one
good from agent~2's bundle changes agent~2's counts to either
\((2-x,5-y)\), if \(x<3\), or \((3-x,4-y)\), if \(y<5\). Thus agent~1 satisfies
the \EFone{} condition exactly when \(g_1(x,y)\) is at least one of these two
values, whenever the corresponding deletion is possible. Symmetrically, agent~2
satisfies the \EFone{} condition exactly when \(g_2(3-x,5-y)\) is at least one
of \(g_2(x-1,y)\), if \(x>0\), and \(g_2(x,y-1)\), if \(y>0\). If the envied
bundle is empty, the condition is vacuous.

The following grid records the outcome of these tests. An entry \(E\) marks a
split satisfying \EFone{} for both agents; an entry \(1\) marks a split where
agent~1 still envies after the deletion of any single good from agent~2's
bundle; and an entry \(2\) marks the symmetric failure for agent~2.
\[
\begin{array}{c|cccccc}
 & y=0 & y=1 & y=2 & y=3 & y=4 & y=5 \\
\hline
x=0 & 1 & 1 & 1 & 1 & E & 2 \\
x=1 & 1 & 1 & 1 & E & 2 & 2 \\
x=2 & 1 & 1 & E & 2 & 2 & 2 \\
x=3 & 1 & E & 2 & 2 & 2 & 2 .
\end{array}
\]
Therefore \EFone{} holds exactly at the four balanced splits
\[
  (0,4),\qquad (1,3),\qquad (2,2),\qquad (3,1),
\]
where each agent receives four goods.

It remains to show that none of these four \EFone{} allocations is weakly
Pareto optimal. For each \EFone{} split, the table below gives the utility pair
\[
  \bigl(g_1(x,y),\,g_2(3-x,5-y)\bigr)
\]
and exhibits another split that strictly improves both agents.
\[
\begin{array}{c|c|c}
\text{\EFone{} split} & \text{utility pair} &
\text{strictly dominating split and utility pair} \\
\hline
(0,4) & (3+3\varepsilon,3+5\varepsilon) &
  (1,2)\text{ with }(3+4\varepsilon,3+6\varepsilon) \\
(1,3) & (3+5\varepsilon,3+3\varepsilon) &
  (0,5)\text{ with }(3+6\varepsilon,3+4\varepsilon) \\
(2,2) & (3+5\varepsilon,3+3\varepsilon) &
  (0,5)\text{ with }(3+6\varepsilon,3+4\varepsilon) \\
(3,1) & (3+\varepsilon,3+5\varepsilon) &
  (1,2)\text{ with }(3+4\varepsilon,3+6\varepsilon).
\end{array}
\]
Thus every \EFone{} allocation is strictly Pareto dominated by one of the
displayed allocations. Consequently, no allocation is both \EFone{} and
\wPO{}, and hence no allocation is both \EFone{} and \PO{}.

The same domination table gives a quantitative strengthening. In every row,
both agents improve by a factor of at least
\[
  \min\left\{
    \frac{3+4\varepsilon}{3+3\varepsilon},
    \frac{3+6\varepsilon}{3+5\varepsilon},
    \frac{3+4\varepsilon}{3+\varepsilon}
  \right\}
  =
  \frac{3+6\varepsilon}{3+5\varepsilon}.
\]
For \(\varepsilon=1/6\), this factor is \(24/23\). Hence, for this choice of
\(\varepsilon\), the instance has no allocation that is both \EFone{} and
\(\alpha\)-\PO for any \(\alpha>23/24\).
\end{proof}

\begin{remark}[Extending the \EFone{}+\PO{} obstruction to \(n\) agents]
The two-agent obstruction can be lifted to any number of agents by a simple
padding argument. This extension gives an obstruction to \EFone{}+\PO{}, but
not to \EFone{}+\wPO{}.
Starting from the two-agent instance in \Cref{thm:submodular-ce}, add one
private good \(p_i\) for each additional agent \(i=3,\ldots,n\). Agents \(1\)
and \(2\) assign zero marginal value to these private goods and retain their
original valuations over the original eight goods. Each agent \(i\ge 3\) only
values her private good, her value is \(1\) for any bundle containing \(p_i\)
and \(0\) otherwise. It is straightforward to see that no \EFone allocation is \PO on this instance.
\end{remark}

The padding argument above cannot rule out the existence of \(\EFone{}\) and \(\wPO{}\) allocations, since the added
agents can already receive their maximum possible value from their private
goods. To obtain a weak-Pareto obstruction for larger instances, we instead run
independent copies of the two-agent construction on disjoint pairs of agents.
This gives the following extension for every even number of agents.

\begin{corollary}
\label{cor:submodular-even-wpo}
For every even \(n\ge2\), there is an \(n\)-agent goods instance with monotone
submodular valuations such that no allocation is both \EFone and \wPO.
\end{corollary}

\subsection{An Obstruction for Nearly Identical Valuations}

We next show that the existence of \EFone{} and \PO{} allocations is surprisingly
fragile. When all agents have identical submodular valuations, the two
requirements are compatible.

\begin{observation}[Identical submodular agents]\label{thm:identical-submodular}
If all agents have the same monotone submodular valuation \(v\), then every
leximin allocation is both \EFone{} and \PO{}.
\end{observation}

For identical monotone valuations with nonzero marginals,
\citet{plaut2020almost} prove the stronger guarantee \EFX{}+\PO{}. The point
of \Cref{thm:identical-submodular} is that, for the weaker requirement
\EFone{}+\PO{} under submodularity, the nonzero-marginal assumption can be
dropped.

The next theorem shows that this positive result is extremely sensitive to even
the smallest asymmetry between agents. We construct a two-agent instance in
which the agents have the same valuation up to the names of two goods. Agent
\(2\)'s valuation is obtained from agent \(1\)'s valuation by swapping the
labels of two items. Thus, the two valuations have exactly the same structure.
They agree on every set that contains neither of the swapped items, and also on
every set that contains both of them; they can differ only on sets that contain
exactly one of the two swapped items. In this sense, the agents disagree only
about which of two item labels plays which role. Nevertheless, \EFone{} and
\PO{} are incompatible.

Valuations that differ only by a relabeling of the items have recently been
shown to create nontrivial difficulties in fair division,
\citet{mackenzie2026counterexamples} show that \EFX{} may fail to exist even in
this highly symmetric setting.

\begin{theorem}[A single swap of item labels can break existence]
\label{thm:permuted-identical}
There are two agents with monotone submodular valuations such that agent \(2\)'s
valuation is obtained from agent \(1\)'s valuation by swapping two item labels,
and yet no allocation is both \EFone{} and \PO{}.
\end{theorem}

Thus, the existence of \EFone{} and \PO{} allocations is not robust, even a single swap of item labels can destroy existence.

\section{Structured Subclasses: Barriers and Positive Results}\label{sec:subclasses}

The previous section shows that submodularity alone does not guarantee the compatibility of \EFone{} and \PO{}. This leaves a more refined question: how much additional structure must a submodular valuation have before the additive compatibility result becomes plausible again? We turn to weighted matroid-rank valuations, a canonical and highly structured subclass of gross substitutes, as a natural frontier.

This section gives two complementary messages. First, even in the common-weight matroid-rank setting, the strongest efficiency requirements are too ambitious. \EFone{} can be incompatible with \fPO{}, or equivalently, every allocation that is maximizes some positive weighted-welfare violate \EFone. 
Thus the standard market-style or weighted-welfare route cannot prove the desired result in a black-box way. Second, this obstruction is methodological rather than existential for the structured cases we identify. We introduce a \emph{common-envelope condition} and show that every instance satisfying this condition admits an allocation that is both \EFone{} and \PO{}. Since common-weight matroid-rank valuations satisfy the common-envelope condition, this yields a new positive result for that class. We therefore conjecture that the \EFonePO{} guarantee extends to arbitrary \WMR valuations, even though the stronger \EFonefPO{} target fails.

\begin{conjecture}[\WMR frontier]\label{conj:wmr-ef1-po}
Every \WMR valuation instance admits an allocation that is both \EFone and \PO.
\end{conjecture}

The rest of the section develops this frontier. We first present the fractional-efficiency obstruction, and then prove the common-envelope positive result.

\subsection{Why Fractional-Efficiency Methods Are Insufficient}

By \Cref{prop:fpo-welfare}, \fPO allocations are exactly the deterministic
allocations supported by strictly positive weighted-welfare maximizers.  Thus a
proof that always selects an \EFone allocation by maximizing a positive weighted
welfare, or by first proving fractional Pareto optimality, would prove
\EFone{}+\fPO.  The following two results show that this target is already too
strong.

\begin{theorem}[\CWMR failure for \EFone+ \fPO]\label{thm:oxs-fpo-ce}
There is a two-agent \CWMR instance in which every \EFone allocation fails \fPO.  Moreover, viewed as a profile of set
functions, the obstruction is robust: there is \(\delta_0>0\) such that every
profile \((\tilde v_1,\tilde v_2)\) with
\[
  \max_{i\in\{1,2\}}\max_{S\subseteq M}
  |\tilde v_i(S)-v_i(S)|<\delta_0
\]
also has no allocation that is both \EFone and \fPO.
\end{theorem}


The theorem shows that the counterexample is robust to perturbations. Hence, the non-existence result is not merely an artifact of a degenerate instance; rather, it reflects a genuine obstruction that persists on a set of generic instances.

We also show that \EFonefPO{} can fail for additive mixed manna, even under the
standard mixed-manna convention for \EFone{}.  Thus, weighted-welfare
maximization can be too strong a target in signed allocation settings as well.

\begin{theorem}[Robust mixed-manna failure for \fPO]\label{thm:mixed-fpo-ce}
There is a three-agent, four-item additive mixed-manna profile \( ( v_1,v_2,v_3 )\) and a
constant \(\delta_0>0\) such that every additive valuation profile \(( \tilde v_1,\tilde v_2,\tilde v_3 )\) satisfying
\[
  \max_{i\in N}\max_{o\in M}
  |\tilde v_i(o)-v_i(o)|<\delta_0
\]
has no allocation that is both \EFone{} and \fPO{}.  Moreover, the same fixed
allocation is \EFonePO{} throughout this neighborhood.
\end{theorem}



\subsection{Common Envelopes}\label{sec:common-envelope}

The positive result in this section isolates a setting in which agents may
differ in their feasibility constraints while still evaluating feasible sets
on a common scale.  Let \(g:2^M\to\RR_{\ge0}\) be a normalized monotone set
function, which we call the \emph{envelope}.  A family
\(\cC\subseteq 2^M\) is \emph{downward closed} if
\(A\in\cC\) and \(B\subseteq A\) imply \(B\in\cC\).  Each agent \(i\) has a
downward-closed family \(\cC_i\subseteq 2^M\) of \emph{certificates}, these
are the sets that agent \(i\) can realize at their full envelope value.  The
value of a bundle is the envelope value of the best certificate contained in it:
\[
  v_i(S)=\max\{g(T):T\subseteq S,\ T\in\cC_i\}.
\]

We focus on the class of instances for which each valuation \(v_i\) defined
in this way is submodular.  This class is already nontrivial, it contains the
instance of \Cref{thm:oxs-fpo-ce}, and therefore \EFone{}+\fPO can fail.
Plain leximin is not by itself an adequate selection rule either, there exists an instance where a leximin-optimal allocation violates \EFone.  Nevertheless, we show
that with the right tie-breaking one can choose an allocation that is still
leximin-optimal among full allocations.  The proof first chooses disjoint
certificates, viewed as a partial allocation, whose envelope values are
lexicographically maximum, and only then assigns the residual goods in a
careful way.

\begin{theorem}\label{thm:common-envelope}
Let \(g:2^M\to\RR_{\ge0}\) be monotone.  For each agent
\(i\in N\), let \(\cC_i\subseteq 2^M\) be a downward-closed family, and define
\[
  v_i(S)=\max\{g(T):T\subseteq S,\ T\in\cC_i\}.
\]
If each \(v_i\) is submodular, then there is an allocation that is both
\EFone and \PO{}.  Moreover, such an allocation can be chosen to be leximin-optimal.
\end{theorem}

\begin{proof}
Call a partial allocation \(P=(P_1,\ldots,P_n)\) a  \textit{certificate} allocation if
the bundles are disjoint and \(P_i\in\cC_i\) for every \(i\). Among all such certificate allocations, let \(P^*\) be a leximin-optimal one. 
That is, \(P^*\) is  a certificate allocation  which maximizes the value given to the worst-off agent, and subject to that, maximizes the value of the second-worst-off agent, and so on. Note that as each $P^*_i\in \cC_i$, the agents' values are given by the envelope function \(g\) evaluated on their certificate,  \( v_i(P^*_i) = g(P^*_i) \).

\bigskip

\noindent\textit{\(P^*\) is satisfies \EFone.} Suppose on the contrary that \(P^*\) is not \EFone. Then there exist agents \(i\) and \(j\) such that agent \(i\) envies agent \(j\)'s bundle even after the removal of any single good from it. This means that for every good \(e \in P^*_j\), we have \(v_i(P^*_j\setminus\{e\})>v_i(P^*_i)\). 

First, observe that there must exist some good \(e\in P^*_j\) that has positive marginal value for agent \(i\). If not, by submodularity, we have
\begin{align*}
  0 &= \sum_{e\in P^*_j} v_i (P^*_i \cup \{e\}) -v_i (P^*_i) \\ 
    & \geq  v_i(P^*_i\cup P^*_j) - v_i(P^*_i). 
\end{align*}
By monotonicity, we have \(v_i(P^*_i\cup P^*_j) \geq v_i(P^*_i)\), and hence $v_i(P^*_i\cup P^*_j) = v_i(P^*_i)$. But this contradicts the assumption that agent $i$ envied agent $j$. Thus, there must exist a good $ e' \in P^*_j$ with positive marginal value for agent $i$.

Consider such a good \(e'\in P^*_j\) with positive marginal value for agent \(i\). Let \(P'_i \subseteq P^*_i \cup \{e'\}\) be a certificate that attains the value \(v_i(P^*_i \cup \{e'\})\), so that \(g(P'_i) = v_i(P^*_i \cup \{e'\}) > v_i(P^*_i)\). 
Note that the new bundles \(P'_i\) and \(P^*_j\setminus\{e'\}\) are still disjoint, and \(P'_i \in \cC_i\) and \(P^*_j\setminus\{e'\} \in \cC_j\) by downward closure. Thus, we may consider a new certificate allocation \[\widetilde{P}_k = \begin{cases} P'_i  & \text{if } k=i \\ 
P^*_j\setminus\{e'\} & \text{if } k=j \\
P^*_k & \text{otherwise}
  \end{cases}.\]
We now compare the envelope values of the new allocation \(\widetilde{P}\) with those of \(P^*\). Observe that for every agent \(k\notin \{i,j\}\), we have \(v_k(\widetilde{P}_k) = v_k(P^*_k)\). On the other hand \(v_i(P'_i) > v_i(P^*_i)\), and  \(  v_j (P^*_j\setminus\{e'\}) =g(P^*_j\setminus\{e'\}) \ge v_i(P^*_j\setminus\{e'\})>v_i(P^*_i)\), it follows that \(\widetilde{P}\) lexicographically dominates \(P^*\). This contradicts the leximin-optimality of \(P^*\), and hence \(P^*\) satisfies \EFone.

\bigskip

\noindent\textit{Extending to complete allocation.}  Let \(R=M\setminus\bigcup_i P^*_i\) be the residual goods. These goods are worthless on top of the chosen certificates:
\begin{equation}\label{eq:1}
      v_i(P^*_i\cup X)=v_i(P^*_i)
  \qquad
  \text{for every } i \text{ and every } X\subseteq R.
\end{equation}
Indeed, if \(v_i(P^*_i\cup X)>v_i(P^*_i )\), then some certificate
\(T\subseteq P_i\cup X\) has \(g(T)>v_i(P^*_i )\).  Replacing \(P_i\) by \(T\) and
leaving all other certificates unchanged gives a disjoint certificate
allocation that improves \(i\) and hurts no one, contradicting the lexicographical optimality of \(P^*\). 

We now extend the possibly partial allocation \(P^*\) to a complete allocation by assigning the residual goods \(R\) to the agents in a way that preserves \EFone. We iteratively assign the residual goods to the agents, while maintaining a partial allocation \(A=(A_1,\ldots,A_n)\) such that each \(A_i=P^*_i\cup X_i\) for some disjoint \(X_i\subseteq R\). Note that by \eqref{eq:1}, the value of each agent \(i\) remains \(v_i(A_i)=v_i(P^*_i)\) throughout this process. Define the envy graph of \(A\) by putting an edge \(i\to j\) whenever \(v_i(A_j)> v_i(A_i)=v_i(P^*_i)\).

Observe that whenever \(A\) is a partial allocation satisfying \(A_i = P^* \cup X_i\), the envy graph is acyclic.  If there were a directed cycle \(i_1\to i_2\to\cdots\to i_k\to i_1\), then for each \(t\), agent \(i_t\) would have a certificate \(T_{i_t}\subseteq A_{i_{t+1}}\) with
\(g(T_{i_t})=v_{i_t}(A_{i_{t+1}})>v_{i_t}(A_{i_{t}})=v_{i_t}(P^*_i)\), where indices are modulo \(k\).
These certificates are disjoint because the bundles \(A_{i_{t+1}}\) are
disjoint.  Giving each cycle agent \(i_t\) the certificate \(T_{i_t}\), and
letting every agent outside the cycle keep her original \(P^*_i\), gives a
certificate allocation that strictly improves all cycle agents and hurts no
one. This however contradicts the leximin-optimality of \(P^*\).  Hence, the envy graph is acyclic whenever each agent \(i\) receives a bundle of the form $A_i=P^*_i\cup X_i$.

Initialize the partial allocation \(A^0_i=P^*_i\) for all \(i\). The envy graph of \(A^0\) is acyclic, and \(A^0\) is \EFone. We now inductively show that we can assign the residual goods \(R\) to the agents while maintaining \EFone.
At each step, pick a remaining residual good \(e\in R\) and assign it to an agent \(s\) who is the source node in the envy graph. Since \(s\) is a source, adding \(e\) to \(A_s\) does not create \EFone envy from any other agent.  Update the partial allocation to \(A^{t+1}_s=A^t_s\cup \{e\}\) and \(A^{t+1}_i=A^t_i\) for all \(i\neq s\). Assuming the allocation $A^{t}$ preserves $\EFone$, the allocation $A^{t+1}$ preserves $\EFone$.
The envy graph of \(A^{t+1}\) remains acyclic, since the allocation is of the form \(A^{t+1}_i=P^*_i\cup X_i\) for some disjoint \(X_i\subseteq R\),
and the value of each agent remains unchanged, \(v_i(A^{t+1}_i)=v_i(P^*_i)\).   Repeat this process until all residual goods have been assigned. Denote the final allocation by \(\widetilde{A}\).  By construction, \(\widetilde{A}\) is a complete allocation that preserves \EFone and satisfies $v_i (\widetilde{A}_i)=v_i(P^*_i)$.

We now show that \(\widetilde{A}\) is leximin-optimal among all complete allocations.  Suppose not, and let \(B\) be a complete allocation that lexicographically dominates \(\widetilde{A}\).  Then there exists an agent \(i\) such that \(v_i(B_i)>v_i(\widetilde{A}_i)=v_i(P^*_i)\).  This implies that \(B_i\) contains a certificate \(T_i\subseteq B_i\) with \(g(T_i)>g(P^*_i)\).  Let \(Q=(Q_1,\ldots,Q_n)\) be the certificate allocation obtained by replacing \(P^*_i\) with \(T_i\) and leaving all other certificates unchanged.  Then \(Q\) lexicographically dominates \(P^*\), contradicting the leximin-optimality of \(P^*\). Thus, \(\widetilde{A}\) is leximin-optimal among all complete allocations, and hence is \PO.

\end{proof}

\begin{remark}
The common envelope condition is essential, the \EFone step compares agent \(i\)'s
value for \(P^*_j\setminus\{e\}\) with the value \(g(P^*_j\setminus\{e\})\) that
the leximin objective credits to \(j\) after the deletion, and without a
common scale these two numbers are incomparable.
\end{remark}

\begin{corollary}[Common-weight matroid rank]\label{cor:common-weight-matroid}
Let each agent \(i\) have a matroid \(\M_i=(M,\I_i)\) on the common ground set
\(M\).  All agents share nonnegative item weights \(w:M\to\RR_{\ge0}\), and
\[
  v_i(S)=\max\{w(I):I\subseteq S,\ I\in\I_i\}.
\]
Then an \EFonePO{} allocation exists, and can be chosen to be leximin-optimal.
\end{corollary}

\begin{proof}
Apply \Cref{thm:common-envelope} with the additive envelope
\(g(S)=w(S)=\sum_{e\in S}w(e)\) and certificate family \(\cC_i=\I_i\).
It remains to recall that the induced valuation is a weighted matroid-rank
function and is therefore monotone submodular \cite{schrijver2003combinatorial}.
\end{proof}

The common-envelope condition is best viewed as a way of allowing
heterogeneous feasibility constraints while keeping a common value scale.  In
particular, \Cref{cor:common-weight-matroid} includes common-weight rank
valuations for familiar matroid families such as partition, laminar, graphic,
and transversal matroids.  The transversal case can be read as an
approval-based common-weight OXS/assignment model, agents may have different
unit-demand slots or approval graphs, but every approved copy of item \(e\)
contributes the same common weight \(w(e)\). 

Another demonstration of the usefulness of \Cref{thm:common-envelope} is that it shows, for a shared submodular function \(g\) and arbitrary agent-specific sets \(E_i\subseteq M\), the valuation profile
\(
v_i(S)=g(S\cap E_i)
\)
admits an allocation that is both \EFone{} and \PO{} by taking \(\cC_i=2^{E_i}\).

\section{Approximate Pareto Optimality}\label{sec:approx}

The submodular counterexample given in \Cref{thm:submodular-ce}  rules out \EFone and $\alpha$-\PO for any \(\alpha>\frac{23}{24}\). Since every submodular valuation is subadditive, the same impossibility immediately carries over to the broader class of subadditive valuations.
This strengthens the subadditive impossibility of \citet{CKM16}. Their construction shows that \EFone{} and \PO{} need not be compatible under subadditive valuations; however, their instance still admits an allocation that is both \EFone{} and \wPO. 

In this section, we ask how far this incompatibility can be pushed. By the result of \cite{BS2026}, every subadditive instance admits an allocation that is both \EFone{} and \(\tfrac12\)-\PO{}. Thus, the best possible approximation threshold lies between \(\tfrac12\) and \(\frac{23}{24}\). The natural next question is to determine this threshold.
We first show such a sharp threshold for two agents.

\begin{theorem}[Two-agent subadditive barrier]\label{thm:subadditive-alpha}
For every \(\frac{1}{\sqrt{2}}<\alpha\le1\), there is a two-agent subadditive instance with
no allocation that is both \EFone and \(\alpha\)-Pareto optimal.
\end{theorem}

\begin{proof}
Let \(r=1/\sqrt2\) and fix \(\varepsilon>0\) sufficiently small.  Again use
three \(A\)-goods and five \(B\)-goods, and write \(x,y\) for the type counts
in a bundle.  The construction has the same shape as
\Cref{thm:submodular-ce} --- \EFone will hold exactly at the four balanced
splits, and each of those is dominated by an unbalanced one --- but the gap
is now far wider.  The tables below are built almost entirely from the
three values \(\tfrac12\), \(r\), and \(1\), with \(\varepsilon\)-terms only
breaking ties; a value that jumps from \(\tfrac12\) to \(1\) in a single
step violates diminishing returns, and it is exactly these jumps, forbidden
in the submodular world, that widen the gap from \(24/23\) to \(\sqrt2\).
In the limit \(\varepsilon\to0\), each balanced split pays the two agents
\((\tfrac12,r)\) or \((r,\tfrac12)\), while the dominating splits pay
\((r,1)\) and \((1,r)\): both coordinates improve by the factor
\(r/\tfrac12=1/r=\sqrt2\).  The two agents' valuations are \(g_1(x,y)\) and
\(g_2(x,y)\), given by the following tables:
\[
\begin{array}{c|cccccc}
g_1(x,y) & y=0 & y=1 & y=2 & y=3 & y=4 & y=5\\
\hline
x=0 & 0 & \frac12 & \frac12 & \frac12 & r+\varepsilon & r+2\varepsilon\\
x=1 & \frac12 & \frac12 & \frac12 & \frac12+\varepsilon & r+2\varepsilon & 1\\
x=2 & \frac12 & \frac12 & \frac12+\varepsilon & 1 & 1 & 1\\
x=3 & r & r & r & 1 & 1 & 1
\end{array}
\]
and
\[
\begin{array}{c|cccccc}
g_2(x,y) & y=0 & y=1 & y=2 & y=3 & y=4 & y=5\\
\hline
x=0 & 0 & \frac12 & \frac12 & \frac12 & \frac12+\varepsilon & 1\\
x=1 & \frac12 & \frac12 & r & r+\varepsilon & r+2\varepsilon & 1\\
x=2 & \frac12 & \frac12 & r+\varepsilon & r+2\varepsilon & 1 & 1\\
x=3 & \frac12 & \frac12+\varepsilon & r+2\varepsilon & 1 & 1 & 1.
\end{array}
\]
These tables are normalized and monotone for sufficiently small \(\varepsilon\).
Subadditivity reduces, for typed valuations, to the finite capped-union
inequalities
\[
  g_i(\min\{3,x+x'\},\min\{5,y+y'\})
  \le g_i(x,y)+g_i(x',y')
\]
for \(0\le x,x'\le3\) and \(0\le y,y'\le5\).  If one summand is empty the
two sides agree, so assume both summands are nonempty.  At \(\varepsilon=0\)
every nonempty cell is worth at least \(\tfrac12\) and no cell exceeds
\(1\), so the two summands already cover the left side.  It remains to
handle the \(\varepsilon\)-terms.  Where the comparison at \(\varepsilon=0\)
is strict, the slack is at least \(r-\tfrac12\), far larger than any
\(\varepsilon\)-term; where it is tight, the left side is a cell of value
exactly \(1\) and both summands sit at \(\tfrac12\) --- and since every
\(1\)-cell has \(\varepsilon\)-coefficient \(0\) in both tables while the
right side can only gain, the inequality survives.  Hence all the
inequalities hold for all sufficiently small \(\varepsilon>0\).

Checking each split $(x,y)$ and verifying the \EFone{} inequalities for both agents shows that the \EFone{} splits are precisely
\[
  (0,4),(1,3),(2,2),(3,1).
\]
They are dominated as follows, where the displayed utility pairs are the limits
as \(\varepsilon\to0\):
\[
\begin{array}{c|c|c|c}
\text{\EFone split} & \text{limit utility} & \text{dominator} &
\text{dominator limit utility}\\
\hline
(1,3) & (\frac12,r) & (3,0) & (r,1)\\
(2,2) & (\frac12,r) & (3,0) & (r,1)\\
(0,4) & (r,\frac12) & (2,3) & (1,r)\\
(3,1) & (r,\frac12) & (2,3) & (1,r).
\end{array}
\]
The limiting improvement factor in every row is
\[
  \min\left\{\frac{r}{1/2},\frac{1}{r}\right\}=\sqrt2.
\]
Hence, for every \(\eta>0\), choosing \(\varepsilon\) sufficiently small makes
every \EFone allocation strictly improvable for both agents by a factor at
least \(\sqrt2-\eta\).  Given \(1/\sqrt2<\alpha\le1\), choose \(\eta>0\) with
\(\sqrt2-\eta>1/\alpha\).  Then no \EFone allocation is \(\alpha\)-Pareto
optimal.
\end{proof}

Next, we show that the same incompatibility obstruction can be extended to any number of agents. Since such obstructions can become brittle when additional agents are introduced, this extension requires a considerably more delicate argument.

\begin{theorem}[Subadditive barrier for three or more agents]\label{thm:subadditive-alpha-alln}
For every \(n\ge3\) and every \( \frac{1}{\sqrt2}<\alpha\le1\), there is an \(n\)-agent
monotone subadditive instance with no allocation that is both \EFone and
\(\alpha\)-\PO.
\end{theorem}

For two agents, the bound given by \Cref{thm:subadditive-alpha} is tight.  \citet[Theorem A.2]{BS2026}
show that every two-agent subadditive instance has an \EFone allocation whose \NSW is at least \(\frac{1}{\sqrt{2}}\) times that of the optimal, and hence $\frac{1}{\sqrt{2}}$-\PO.

For \(n\geq 3\), a slight modification of the algorithm of \citet{BS2026} shows that every \(n\)-agent subadditive instance admits an \EFone{} allocation whose \NSW is at least a \(2^{-(1-1/n)}\) fraction of the optimum, and hence \( \frac{1}{2^{1-1/n}} \)-\PO{}. Specifically, in their \emph{Complete Set Growing} algorithm, modify the envy-cycle step by giving priority to items in \(U_s\) whenever agent \(s\) is a source in the envy graph. The improved approximation factor follows from the fact that at least one of the unallocated sets \(U_i\) is empty in some iteration of the envy-cycle procedure.

We suspect that the above guarantee \EFone and \( \frac{1}{2^{1-1/n}} \)-\PO{} for subadditive valuations is tight. In particular, it recovers the sharp two-agent threshold when \(n=2\).

\begin{conjecture}[Subadditive approximation threshold]\label{conj:many-agent-alpha}
For every \(n\ge2\), the best universal \(\alpha\) for which every \(n\)-agent
subadditive goods instance admits an \EFone allocation that is
\(\alpha\)-\PO is
\[
  \alpha_n=2^{-(1-1/n)}.
\]
Equivalently, every such instance should admit an \EFone allocation that is
\(\alpha_n\)-\PO, and for every \(\varepsilon>0\) there exists an instance where no allocation is simultaneously \EFone and 
\((\alpha_n+\varepsilon)\)-\PO.
\end{conjecture}

\section*{Acknowledgements}

The paper is supported by the NSF--CSIRO project on ``Fair Sequential
Collective Decision-Making'' and the ARC Laureate Project FL200100204 on
``Trustworthy AI''.


\bibliographystyle{plainnat}
\bibliography{references}

\clearpage
\appendix

\section[Deferred Proofs from Section 3]{Deferred Proofs from Section~\ref{sec:core-submodular}}\label{app:proofs-sec-core-submodular}

\subsection[Proof of Corollary 3.2]{Proof of \Cref{cor:submodular-even-wpo}}\label{app:proof-cor-submodular-even-wpo}

Write \(n=2k\).  If \(k=1\), this is exactly
\Cref{thm:submodular-ce}.  Hence assume \(k\ge2\).

We take \(k\) disjoint copies of the two-agent instance from
\Cref{thm:submodular-ce}, with \(\varepsilon=1/6\).  Copy \(p\) has goods
\[
  M^p=A_p\cup B_p,
  \qquad |A_p|=3,\qquad |B_p|=5,
\]
and two agents, denoted \((p,1)\) and \((p,2)\).  Let \(g_1,g_2\) be the two
count valuations from \Cref{thm:submodular-ce}.  Thus \(g_a(x,y)\) is the
value, for agent \(a\), of a bundle with \(x\) goods of type \(A_p\) and
\(y\) goods of type \(B_p\).

Agent \((p,a)\) values goods in her own copy \(M^p\) according to \(g_a\), and
values goods outside \(M^p\) only by a very small additive term.  Formally,
\[
  V_{p,a}(S)
  =
  g_a\bigl(|S\cap A_p|,|S\cap B_p|\bigr)
  +
  \eta |S\setminus M^p|,
  \qquad
  0<\eta<\frac{1}{48(k-1)}.
\]
There are \(8(k-1)\) goods outside \(M^p\), so all goods outside \(M^p\)
together are worth less than
\[
  8(k-1)\eta<\frac16
\]
to agent \((p,a)\).

This is the whole point of the parameter \(\eta\): a gain of \(1/6\) inside an
agent's own copy is enough to compensate for losing every good she may have
held outside her own copy.

The valuations \(V_{p,a}\) are monotone submodular.  The first term is the
monotone submodular valuation from \Cref{thm:submodular-ce}, applied only to
the goods in \(M^p\).  The second term is nonnegative and additive.  Hence their
sum is monotone submodular.

We will use the following two consequences of \Cref{thm:submodular-ce} at
\(\varepsilon=1/6\).  First, in the two-agent instance, the \EFone allocations
are exactly the splits in which both agents receive four goods.  Equivalently,
if one agent receives at least five of the eight goods, then the other agent
still envies her after any one-good deletion.  Since all table entries are
multiples of \(1/6\), every such strict envy gap is at least \(1/6\).

Second, every one of those four-good/four-good \EFone splits is strictly
dominated by another split of the same eight goods, and both agents gain at
least \(1/6\).

Now let \(X\) be any \EFone allocation.  We will construct another allocation
that strictly improves every agent.

Fix a copy \(p\).  For \(a\in\{1,2\}\), let
\[
  L_{p,a}=X_{(p,a)}\cap M^p
\]
be the part of agent \((p,a)\)'s bundle that lies in her own copy.  We first
show that neither agent in the pair can hold more than four goods from her own
copy:
\[
  |L_{p,1}|\le4
  \qquad\text{and}\qquad
  |L_{p,2}|\le4 .
\]

Suppose, for contradiction, that \(|L_{p,1}|\ge5\).  Look only at the goods of
\(M^p\), and imagine the two-agent split in which agent \(1\) receives
\(L_{p,1}\) and agent \(2\) receives all remaining goods of \(M^p\).  In this
split, agent \(1\) has at least five goods, so by the two-agent theorem agent
\(2\) still envies agent \(1\) after every one-good deletion, with gap at least
\(1/6\).

Agent \((p,2)\)'s actual own-copy bundle \(L_{p,2}\) is contained in
\(M^p\setminus L_{p,1}\).  Therefore, by monotonicity, the same comparison is
still true against \(L_{p,2}\): after deleting any one good from
\(L_{p,1}\), agent \((p,2)\) values the remaining own-copy part at least
\[
  g_2(L_{p,2})+\frac16 .
\]
If instead we delete a good outside \(M^p\) from \(X_{(p,1)}\), the own-copy
part \(L_{p,1}\) does not change, so the remaining bundle is worth at least as
much.  Hence every one-good deletion from \(X_{(p,1)}\) leaves agent \((p,2)\)
with value at least
\[
  g_2(L_{p,2})+\frac16 .
\]

But agent \((p,2)\)'s own value in \(X\) is less than this, because all goods
outside \(M^p\) together are worth less than \(1/6\):
\[
  V_{p,2}(X_{(p,2)})
  \le
  g_2(L_{p,2})+8(k-1)\eta
  <
  g_2(L_{p,2})+\frac16 .
\]
Thus no one-good deletion from \(X_{(p,1)}\) removes agent \((p,2)\)'s envy of
agent \((p,1)\).  This contradicts \EFone of \(X\).  Therefore
\(|L_{p,1}|\le4\).  The same argument, with the two agents interchanged, gives
\(|L_{p,2}|\le4\).

We now use this bound to repair copy \(p\).  The sets \(L_{p,1}\) and
\(L_{p,2}\) are disjoint subsets of \(M^p\), and each has size at most four.
Therefore we can add the remaining goods of \(M^p\) to the two agents so that
agent \((p,1)\) has a set \(D_{p,1}\supseteq L_{p,1}\), agent \((p,2)\) has a
set \(D_{p,2}\supseteq L_{p,2}\), and
\[
  D_{p,1}\cup D_{p,2}=M^p,
  \qquad
  |D_{p,1}|=|D_{p,2}|=4.
\]
This gives one of the \EFone splits of the original two-agent instance on
copy \(p\).

By \Cref{thm:submodular-ce}, this four-good/four-good split is strictly
dominated by another split of \(M^p\), say
\[
  Z_{p,1}\cup Z_{p,2}=M^p,
  \qquad
  Z_{p,1}\cap Z_{p,2}=\emptyset,
\]
such that both agents gain at least \(1/6\):
\[
  g_1(Z_{p,1})\ge g_1(D_{p,1})+\frac16,
  \qquad
  g_2(Z_{p,2})\ge g_2(D_{p,2})+\frac16.
\]
Since \(D_{p,a}\supseteq L_{p,a}\), monotonicity gives
\[
  g_a(Z_{p,a})
  \ge
  g_a(L_{p,a})+\frac16
  \qquad\text{for }a=1,2.
\]

Do this independently for every copy \(p\).  Define a new allocation \(Y\) by
giving \(Z_{p,1}\) to agent \((p,1)\) and \(Z_{p,2}\) to agent \((p,2)\), for
each \(p=1,\ldots,k\).  Since the copies \(M^1,\ldots,M^k\) are disjoint, this
allocates every good exactly once.

Fix an agent \((p,a)\).  In \(Y\), she receives only goods from her own copy, so
\[
  V_{p,a}(Y_{(p,a)})
  =
  g_a(Z_{p,a})
  \ge
  g_a(L_{p,a})+\frac16 .
\]
In \(X\), her own-copy part is \(L_{p,a}\), and all goods outside her own copy
are worth less than \(1/6\) in total.  Hence
\[
  V_{p,a}(X_{(p,a)})
  \le
  g_a(L_{p,a})+8(k-1)\eta
  <
  g_a(L_{p,a})+\frac16 .
\]
Therefore
\[
  V_{p,a}(Y_{(p,a)})
  >
  V_{p,a}(X_{(p,a)})
\]
for every \(p\) and every \(a\in\{1,2\}\).

Thus \(Y\) strictly improves every agent relative to \(X\).  Since \(X\) was an
arbitrary \EFone allocation, no \EFone allocation is weakly Pareto optimal.

\subsection[Proof of Observation 3.3]{Proof of \Cref{thm:identical-submodular}}\label{app:proof-thm-identical-submodular}

Let \(A=(A_1,...,A_n)\) be a leximin optimal allocation. Leximin optimality implies
\PO, any Pareto improvement weakly increases every coordinate of the utility
vector and strictly increases at least one, so the sorted utility vector
lexicographically improves.

It remains to prove that $A$ is \EFone.  
Suppose, towards a contradiction, that \(A\) does not satisfy \EFone. 
 Then there exist agents \(i\) and \(j\) such that agent \(i\) envies agent \(j\)'s bundle even after the removal of any single good from it. This means that for every good \(g \in A_j\), we have \(v(A_j\setminus\{g\})>v(A_i)\). 

We first claim that  some good \(g'\in A_j\) has a strictly positive marginal with respect to \(A_i\). If not, by submodularity, we have
\begin{align*}
  0 &= \sum_{g\in A_j} v (A_i \cup \{g\}) -v (A_i) \\ 
    & \geq  v(A_i\cup A_j) - v(A_i). 
\end{align*}
By monotonicity, we have \(v(A_i\cup A_j) \geq v(A_i)\), and hence $v(A_i\cup A_j) = v(A_i)$. But this contradicts the assumption that agent $i$ envied agent $j$. Therefore, there must exist a good $ g'\in A_j$ with positive marginal value for agent $i$.

Transfer $g'$ from agent $j$ to agent $i$, and let $\widetilde{A}$ denote the resulting allocation:
\[\widetilde{A}_k = \begin{cases} A_i \cup \{ g' \}  & \text{if } k=i \\ 
A_j\setminus\{g'\} & \text{if } k=j \\
A_k & \text{otherwise}.
  \end{cases}\]
By the choice of $g'$, we have \( v(\widetilde{A}_i) = v(A_i   \cup \{ g' \} ) >v(A_i  )  \). Moreover, the assumed violation of \EFone gives \( v(\widetilde{A}_j) = v(A_j   \setminus \{ g' \} ) >v(A_i  )  \). The utilities of all other agents remain unchanged. However, this implies $\widetilde{A}$ lexicographically  dominates $A$, contradicting leximin optimality of $A$. 

Therefore, $A$ is both \EFone and \PO.

\subsection[Proof of Theorem 3.4]{Proof of \Cref{thm:permuted-identical}}\label{app:symmetry-certificates}
 Let the goods be
\(a,b,c,d,e,f\).  Agent 1 has valuation \(v\), and agent 2 has the same
valuation after swapping the labels \(d\) and \(e\).  The valuation \(v\) is
specified as follows.  The singleton values are
\[
\begin{array}{c|rrrrrr}
S & a&b&c&d&e&f\\ \hline
v(S)&10&9&8&10&7&6.
\end{array}
\]
The pair values are
\[
\begin{array}{c|rrrrrr}
v(\cdot) & a&b&c&d&e&f\\ \hline
a&-&17&16&20&17&14\\
b&-&-&15&14&16&15\\
c&-&-&-&14&15&14\\
d&-&-&-&-&14&14\\
e&-&-&-&-&-&13
\end{array}
\]
The triple values are
\[
\begin{array}{rrrrr}
abc:19&abd:20&abe:20&abf:19&acd:22\\
ace:19&acf:19&ade:20&adf:22&aef:19\\
bcd:18&bce:22&bcf:21&bde:18&bdf:18\\
bef:22&cde:18&cdf:18&cef:21&def:18.
\end{array}
\]
Every four-item set has value \(22\), except
\[
  v(abcf)=v(acef)=21,
  \qquad v(abde)=20.
\]
Every five-item set and the six-item set have value \(22\), and
\(v(\emptyset)=0\).
A direct marginal check shows that \(v\) is monotone submodular.

The two Pareto optimal utility pairs are
\[
\begin{array}{c|c|c}
S & M\setminus S & (v_1(S),v_2(M\setminus S))\\ \hline
ad & bcef & (20,22)\\
bcdf & ae & (22,20).
\end{array}
\]
Neither allocation is \EFone.  At \(ad\mid bcef\), agent 1 envies agent 2 even
after any one deletion from \(bcef\): the four possible deleted bundles have
values at least \(21>20\).  At \(bcdf\mid ae\), the symmetric statement holds
for agent 2 after applying the transposition.  Thus every \PO allocation fails
\EFone, proving the theorem.

\subsection{Submodularity Marginal Grids}\label{app:marginals}

For the count-table instance in \Cref{thm:submodular-ce}, the \(A\)-marginals
\(\Delta_A g(x,y)=g(x+1,y)-g(x,y)\) and \(B\)-marginals
\(\Delta_B g(x,y)=g(x,y+1)-g(x,y)\) are as follows.
\[
\begin{array}{c|rrrrrr}
\Delta_A g_1 & y=0&1&2&3&4&5\\ \hline
x=0 & 2&2&1+4\varepsilon&5\varepsilon&3\varepsilon&0\\
x=1 & \varepsilon&\varepsilon&\varepsilon&\varepsilon&0&0\\
x=2 & 0&0&0&0&0&0
\end{array}
\qquad
\begin{array}{c|rrrrr}
\Delta_B g_1 & y=0&1&2&3&4\\ \hline
x=0 & 1&1&1&3\varepsilon&3\varepsilon\\
x=1 & 1&4\varepsilon&\varepsilon&\varepsilon&0\\
x=2 & 1&4\varepsilon&\varepsilon&0&0\\
x=3 & 1&4\varepsilon&\varepsilon&0&0.
\end{array}
\]
\[
\begin{array}{c|rrrrrr}
\Delta_A g_2 & y=0&1&2&3&4&5\\ \hline
x=0 & 1+2\varepsilon&1+\varepsilon&1&3\varepsilon&\varepsilon&\varepsilon\\
x=1 & 1+\varepsilon&1-\varepsilon&3\varepsilon&3\varepsilon&0&0\\
x=2 & 1+\varepsilon&5\varepsilon&3\varepsilon&0&0&0
\end{array}
\qquad
\begin{array}{c|rrrrr}
\Delta_B g_2 & y=0&1&2&3&4\\ \hline
x=0 & 1&1&1&5\varepsilon&0\\
x=1 & 1-\varepsilon&1-\varepsilon&3\varepsilon&3\varepsilon&0\\
x=2 & 1-3\varepsilon&3\varepsilon&3\varepsilon&0&0\\
x=3 & \varepsilon&\varepsilon&0&0&0.
\end{array}
\]
Since \(0<\varepsilon\le1/6\), all entries are nonnegative and each row and
column is nonincreasing.  This proves the required diminishing returns.

\section[Deferred Proofs from Section 4]{Deferred Proofs from Section~\ref{sec:subclasses}}\label{app:proofs-sec-subclasses}

\subsection[Proof of Theorem 4.1]{Proof of \Cref{thm:oxs-fpo-ce}}\label{app:oxs-fpo-proof}


There are two agents and four goods
\( M=\{ a,b,c,d \} \), with common weights
\[
  w(a)=w(b)=1,\qquad w(c)=w(d)=3.
\]
Agent 1 has the partition matroid with blocks \(\{a,c\}\) and \(\{b,d\}\),
capacity one in each block.  Agent 2 has the partition matroid with blocks
\(\{a,d\}\) and \(\{b,c\}\), again with capacity one in each block.  Thus
\[
  v_1(S)=\max w(S\cap\{a,c\})+\max w(S\cap\{b,d\}),
\]
and
\[
  v_2(S)=\max w(S\cap\{a,d\})+\max w(S\cap\{b,c\}),
\]
where an empty maximum is \(0\).  This is common-weight weighted matroid rank.

Expanding the table of values gives
\[
\begin{array}{c|rrrrrrrrrrrrrrrr}
S&\emptyset&a&b&c&d&ab&ac&ad&bc&bd&cd&abc&abd&acd&bcd&abcd\\ \hline
v_1(S)&0&1&1&3&3&2&3&4&4&3&6&4&4&6&6&6\\
v_2(S)&0&1&1&3&3&2&4&3&3&4&6&4&4&6&6&6
\end{array}
\]
Write an allocation by the bundle \(A\) given to agent 1, with agent 2 receiving
\(M\setminus A\).  The utility vectors are
\[
\begin{array}{c|rrrrrrrr}
A&\emptyset&a&b&c&d&ab&ac&ad\\ \hline
(v_1(A),v_2(M\setminus A))&(0,6)&(1,6)&(1,6)&(3,4)&(3,4)&(2,6)&(3,4)&(4,3)\\[2mm]
A&bc&bd&cd&abc&abd&acd&bcd&abcd\\ \hline
(v_1(A),v_2(M\setminus A))&(4,3)&(3,4)&(6,2)&(4,3)&(4,3)&(6,1)&(6,1)&(6,0).
\end{array}
\]
A direct EF1 check shows that the \EFone allocations are exactly those with
utility \((3,4)\) or \((4,3)\).  These allocations are ordinary \PO{}.

It remains to identify the \fPO allocations.  By \Cref{prop:fpo-welfare}, an
allocation is \fPO only if it maximizes
\(\lambda_1v_1+\lambda_2v_2\) for some \(\lambda_1,\lambda_2>0\).  Put
\(r=\lambda_1/\lambda_2\).  The objective is \(r u_1+u_2\).  If \(r<1\), the
unique maximizing utility vector is \((2,6)\), attained by
\(A=\{a,b\}\).  If \(r>1\), the unique maximizing utility vector is \((6,2)\),
attained by \(A=\{c,d\}\).  If \(r=1\), these two allocations tie with total
utility \(8\), while every other utility vector has total utility at most \(7\).
Thus the only \fPO allocations are
\[
  \{a,b\}\mid\{c,d\}
  \qquad\text{and}\qquad
  \{c,d\}\mid\{a,b\}.
\]
Both fail \EFone.  In the first allocation, agent 1 has value \(2\), but values
agent 2's bundle \(\{c,d\}\) at \(6\); deleting either \(c\) or \(d\) leaves a
single good of value \(3>2\).  The second allocation fails symmetrically for
agent 2.  Hence no allocation is both \EFone and \fPO, even though \EFonePO{}
allocations exist.

The obstruction is robust as a set-function profile.  Consider any perturbation
\((\tilde v_1,\tilde v_2)\) with
\[
  \max_{i,S}|\tilde v_i(S)-v_i(S)|<\delta
\]
for some \(\delta<1/4\).  In the displayed instance, every allocation that is
not \EFone fails \EFone with margin at least \(1\): for one of the two agents,
her own value is at least \(1\) below the value she assigns to the other
agent's bundle after deleting any single good.  Under the perturbation, each
side of each such comparison changes by less than \(\delta\), so the gap
remains positive whenever \(\delta<1/2\).  Thus every \EFone allocation of the
perturbed profile must be one of the allocations whose unperturbed utility is
\((3,4)\) or \((4,3)\).

It remains to see that these allocations remain fractionally dominated.  The
two allocations \(\{a,b\}\mid\{c,d\}\) and
\(\{c,d\}\mid\{a,b\}\) have unperturbed utilities \((2,6)\) and \((6,2)\).
If an \EFone allocation has utility \((3,4)\), then the lottery putting weights
\(5/8\) and \(3/8\) on these two allocations gives utility
\((7/2,9/2)\), improving both agents by \(1/2\).  If an \EFone allocation has
utility \((4,3)\), the lottery with weights \(3/8\) and \(5/8\) gives
\((9/2,7/2)\), again improving both agents by \(1/2\).  After perturbation,
the value of the fixed lottery for each agent changes by less than \(\delta\),
and the value of the incumbent allocation also changes by less than
\(\delta\).  The domination slack is therefore still at least
\(1/2-2\delta>0\).  Hence every perturbed \EFone allocation is dominated by a
lottery, so no perturbed allocation is \fPO.  Taking any
\(\delta_0<1/4\) proves the robustness claim.

\subsection[Proof of Theorem 4.2]{Proof of \Cref{thm:mixed-fpo-ce}}\label{app:mixed-fpo-proof}
Recall that, in the presence of both goods and chores,  \EFone{} is defined as follows. For every ordered pair \(i,j\),  envy that 
\(i\) has towards \(j\) can be eliminated by deleting one item either from \(j\)'s bundle or
from \(i\)'s own bundle.

For \(\varepsilon>0\), consider three agents \(a_1,a_2,a_3\) and four labelled
items
\[
  M=\{g,c_1,c_2,C\},
\]
with additive singleton values
\[
\begin{array}{c|rrrr}
 & g & c_1 & c_2 & C\\ \hline
a_1 & 3+\varepsilon & -1 & -1-\varepsilon & -4+\varepsilon\\
a_2 & 5+\varepsilon & -2 & -2-2\varepsilon & -5+\varepsilon\\
a_3 & 5-\varepsilon & -2-\varepsilon & -2+\varepsilon & -5.
\end{array}
\tag{1}
\]
We denote this profile by \(v^\varepsilon\).  For
\(0<\varepsilon<1/10\), all twelve displayed singleton values are distinct.

We first study the limiting profile \(v^0\), obtained by setting
\(\varepsilon=0\):
\[
\begin{array}{c|rrrr}
 & g & c_1 & c_2 & C\\ \hline
a_1 & 3 & -1 & -1 & -4\\
a_2 & 5 & -2 & -2 & -5\\
a_3 & 5 & -2 & -2 & -5.
\end{array}
\tag{2}
\]

Let \(h\) be the holder of \(g\).  Suppose that another agent \(r\) holds
\(C\).  Even when both \(c\)-chores are assigned to \(h\), the envy gap of
\(r\) toward \(h\) is at least
\[
  v^0_r(g)+|v^0_r(C)|-2|v^0_r(c_1)|.
\]
This lower bound equals \(5\) when \(r=a_1\) and \(6\) when
\(r\in\{a_2,a_3\}\).  Deleting a single item can reduce the gap by at most
\(4\) for \(a_1\) and by at most \(5\) for \(a_2\) or \(a_3\).  Thus the envy
persists with a margin of at least \(1\).  Consequently, every \EFone{}
allocation assigns \(g\) and \(C\) to the same agent.

Suppose next that \(h\) holds both \(g\) and \(C\), while one of the two
nonholders receives both \(c\)-chores.  Even after deleting either one of her
chores, that agent still envies the remaining agent's empty bundle, again by a
margin of at least \(1\).  Hence every \EFone{} allocation satisfies
\[
  \text{\(g\) and \(C\) have a common holder, and neither nonholder
  receives both \(c\)-chores.}
  \tag{3}
\]
We refer to allocations satisfying~(3) as candidate allocations.

We next show that no candidate allocation can maximize a weighted-welfare
objective with a positive weight vector.  Let
\(\lambda_1,\lambda_2,\lambda_3\geq 0\) be the weights of
\(a_1,a_2,a_3\), respectively.  By additivity, if an item is assigned to an
agent in a weighted-welfare-maximizing allocation, that agent's weighted value
for the item must be at least that of every other agent.

Suppose first that \(a_1\) holds \(g\) and \(C\).  Comparing the weighted
values of these items for \(a_1\) and \(a_2\) gives
\[
  3\lambda_1\ge 5\lambda_2
  \qquad\text{and}\qquad
  -4\lambda_1\ge -5\lambda_2.
\]
Therefore
\[
  4\lambda_1\le 5\lambda_2\le 3\lambda_1,
\]
so \(\lambda_1=\lambda_2=0\).  Comparing the weighted value of \(g\) for
\(a_1\) and \(a_3\) then gives \(3\lambda_1\geq5\lambda_3\), and hence
\(\lambda_3=0\).

Suppose instead that \(a_2\) holds \(g\) and \(C\).  Comparing the weighted
values of these items for \(a_2\) and \(a_3\) gives
\[
  5\lambda_2\ge5\lambda_3
  \qquad\text{and}\qquad
  -5\lambda_2\ge-5\lambda_3,
\]
so \(\lambda_2=\lambda_3\).  Comparing the weighted value of \(g\) for
\(a_2\) and \(a_1\) gives
\[
  5\lambda_2\ge3\lambda_1.
  \tag{4}
\]
Condition~(3) ensures that at least one \(c\)-chore is held by \(a_2\) or
\(a_3\).  If \(a_2\) holds such a chore, comparison with \(a_1\) gives
\[
  -2\lambda_2\ge-\lambda_1,
\]
and hence \(\lambda_1\ge2\lambda_2\).  If \(a_3\) holds it, the same
comparison gives \(\lambda_1\ge2\lambda_3=2\lambda_2\).  In either
case,~(4) implies
\[
  5\lambda_2\ge3\lambda_1\ge6\lambda_2.
\]
Thus \(\lambda_2=0\), and consequently \(\lambda_1=\lambda_3=0\).  The case
in which \(a_3\) holds \(g\) and \(C\) is symmetric.

We have proved that no candidate allocation is supported by a positive weight welfare function.  This stronger conclusion makes the obstruction
robust.  Suppose, for contradiction, that every neighborhood of \(v^0\)
contains an additive profile admitting an allocation that is both \EFone{} and
\fPO{}.  There would then be a sequence of additive profiles \(v^k\) converging
to \(v^0\), together with allocations \(X^k\) that are \EFone{} and \fPO{} at
\(v^k\).  Because there are only finitely many allocations, we may pass to a
subsequence on which \(X^k=X\) is fixed.  The non-\EFone{} allocations ruled
out above violate \EFone{} at \(v^0\) with a strict margin.  By continuity,
none of them can be \EFone{} at all sufficiently large \(v^k\).  Thus \(X\)
must satisfy condition~(3).

By \Cref{prop:fpo-welfare}, for every \(k\) there is a strictly positive
weight vector supporting \(X\) at \(v^k\).  Normalize each vector so that its
coordinates sum to \(1\).  The unit simplex is compact, so a subsequence
converges to a positive vector.  Passing to the limit in the
itemwise weighted-welfare inequalities shows that the limiting vector supports
\(X\) at \(v^0\), contradicting the preceding argument.

It follows that there is some \(\eta>0\) such that every additive profile
\(v\) satisfying
\[
  \lVert v-v^0\rVert_\infty<\eta
  \tag{5}
\]
has no allocation that is both \EFone{} and \fPO{}.  Here and below, the norm
is taken over the singleton-value table.

We now choose
\[
  0<\varepsilon<\min\left\{\frac{\eta}{3},\frac{1}{10}\right\}.
\]
The profile in~(1) satisfies
\[
  \lVert v^\varepsilon-v^0\rVert_\infty=2\varepsilon.
\]
Consequently, every additive profile \(\tilde v\) with
\[
  \lVert\tilde v-v^\varepsilon\rVert_\infty<\varepsilon
  \tag{6}
\]
also satisfies \(\lVert\tilde v-v^0\rVert_\infty<3\varepsilon<\eta\).
Therefore no profile in the open neighborhood~(6) admits an allocation that
is both \EFone{} and \fPO{}.

It remains to show that the same neighborhood admits an \EFonePO{}
allocation.  Give every item to \(a_2\):
\[
  X_{a_1}=X_{a_3}=\emptyset,
  \qquad
  X_{a_2}=M.
  \tag{7}
\]
At \(v^\varepsilon\),
\[
  v^\varepsilon_{a_1}(M)=-3+\varepsilon,
  \qquad
  v^\varepsilon_{a_3}(M)=-4-\varepsilon,
\]
so neither \(a_1\) nor \(a_3\) envies \(a_2\).  Agent \(a_2\)'s envy of
either empty bundle is eliminated by deleting \(C\), since
\[
  v^\varepsilon_{a_2}(M\setminus\{C\})=1-\varepsilon>0.
\]
For every \(\tilde v\) in the neighborhood~(6),
\[
  \tilde v_{a_1}(M)<-3+5\varepsilon<0,
  \qquad
  \tilde v_{a_3}(M)<-4+3\varepsilon<0,
\]
and
\[
  \tilde v_{a_2}(M\setminus\{C\})>1-4\varepsilon>0.
\]
Thus \(X\) remains \EFone{} throughout the neighborhood.

Finally, suppose that another allocation Pareto-dominates \(X\).  Since
\(c_1,c_2,C\) remain chores for \(a_1\) and \(a_3\), any nonempty bundle
assigned to either agent must contain \(g\).  As there is only one good, a
Pareto improvement must transfer a set \(S\ni g\) from \(a_2\) to exactly one
agent \(r\in\{a_1,a_3\}\).  By additivity, this would require
\[
  \tilde v_{a_2}(S)\le0
  \qquad\text{and}\qquad
  \tilde v_r(S)\ge0.
  \tag{8}
\]

At \(v^\varepsilon\), the only sets \(S\ni g\) that \(a_2\) values
nonpositively are
\[
  \{g,c_1,C\},\qquad
  \{g,c_2,C\},\qquad
  M.
  \tag{9}
\]
The only tight case outside this list is
\[
  v^\varepsilon_{a_2}(\{g,C\})=2\varepsilon.
\]
Because \(\{g,C\}\) contains two items, a perturbation satisfying~(6) changes
its value by strictly less than \(2\varepsilon\), so it remains positive.
Every other set \(S\ni g\) outside~(9) has value at least
\(1-\varepsilon\) at \(v^\varepsilon\) and contains at most three items.
Hence its value after perturbation is greater than
\(1-4\varepsilon>0\).

For each set in~(9), both possible recipients value the set strictly
negatively.  Indeed, the largest such value at \(v^\varepsilon\) is
\[
  -2+2\varepsilon<0.
\]
Each set contains at most four items, so throughout~(6) all these values are
strictly less than \(-2+6\varepsilon<0\).  Hence no set \(S\) can satisfy both
inequalities in~(8), and \(X\) is Pareto optimal.

Thus every profile in the open neighborhood~(6) has no allocation that is
both \EFone{} and \fPO{}, while the fixed allocation~(7) remains
\EFonePO{}.  Taking \(\delta_0=\varepsilon\), the profile \(v^\varepsilon\)
proves \Cref{thm:mixed-fpo-ce}.

\section[Deferred Proofs from Section 5]{Deferred Proofs from Section~\ref{sec:approx}}\label{app:proofs-sec-approx}

\subsection[Proof of Theorem 5.2]{Proof of \Cref{thm:subadditive-alpha-alln}}\label{app:proof-thm-subadditive-alpha-alln}

Let \(r=1/\sqrt2\), fix \(n\ge3\), and fix
\(\alpha\in(r,1]\).  The proof is an extension of the two-agent construction
from \Cref{thm:subadditive-alpha}.  Agents \(1\) and \(2\) will play the role
of the two original agents.  The remaining agents will value only auxiliary
goods.  The auxiliary goods are added in a way that creates the following
tension: an \EFone allocation cannot give too many auxiliary goods to the
auxiliary agents, but the efficient allocation will give exactly that many
auxiliary goods to them.

Choose \(0<\varepsilon\) so that the seed valuations \(g_1,g_2\) from
\Cref{thm:subadditive-alpha} are subadditive.  Then choose \(\delta>0\) such
that
\[
  0<4\varepsilon<\delta<r-\frac12
\]
and
\[
  \rho
  :=
  \min\left\{
    \frac{r}{\frac12+\delta},
    \frac{1}{r+2\varepsilon},
    \frac1r
  \right\}
  >
  \frac1\alpha .
\]
This is possible because the displayed minimum tends to
\(\sqrt2=1/r\) as \(\varepsilon,\delta\to0\), while
\(\alpha>r=1/\sqrt2\).  Set
\[
  t=\frac12+\delta .
\]
Then
\[
  \frac12+\varepsilon<t<r .
\]
Thus \(t\) sits between the low seed values and the value \(r\).  In
particular, in the seed tables of \Cref{thm:subadditive-alpha}, a seed value is
above \(t\) exactly when it is at least \(r\).

Let \(m=n-2\), and add \(m\) auxiliary agents, namely agents
\(3,\ldots,n\).  Start with the eight seed goods \(A^3\cup B^5\) from
\Cref{thm:subadditive-alpha}.  Add a set \(C\) of \(mT\) auxiliary goods,
where
\[
  T\ge \max\{3,m+1\}.
\]
For a bundle \(S\), write
\[
  x(S)=|S\cap A|,
  \qquad
  y(S)=|S\cap B|,
  \qquad
  z(S)=|S\cap C|.
\]
We call \((x(S),y(S))\) the seed count vector of \(S\).  If
\(c=(x,y)\), we write \(g_i(c)\) for \(g_i(x,y)\).  A one-good seed deletion
from \(c\) means one of the available vectors
\[
  (x-1,y)\quad\text{if }x>0,
  \qquad
  (x,y-1)\quad\text{if }y>0.
\]

Define
\[
h_T(0)=0,\qquad
h_T(z)=\frac t2\ \text{ for }1\le z\le T-2,\qquad
h_T(z)=t\ \text{ for }z\ge T-1,
\]
and
\[
q_T(0)=0,\qquad
q_T(z)=\frac12\ \text{ for }1\le z\le T-2,\qquad
q_T(T-1)=r,\qquad
q_T(z)=1\ \text{ for }z\ge T.
\]
Agents \(1\) and \(2\) have valuations
\[
  v_1(S)=\max\{g_1(x(S),y(S)),h_T(z(S))\},
  \qquad
  v_2(S)=\max\{g_2(x(S),y(S)),h_T(z(S))\}.
\]
Each auxiliary agent \(k=3,\ldots,n\) has valuation
\[
  v_k(S)=q_T(z(S)).
\]

The interpretation is as follows.  For agents \(1\) and \(2\), auxiliary goods
can raise a low seed value only up to \(t\), which is still below \(r\).  For an
auxiliary agent, \(T-1\) auxiliary goods are worth \(r\), while \(T\) auxiliary
goods are worth \(1\).  So the efficient outcome will want to give \(T\)
auxiliary goods to every auxiliary agent, but \EFone will prevent this.

We first check subadditivity.  The functions \(h_T\) and \(q_T\) are normalized,
monotone, and subadditive.  If \(a,b>0\), then
\[
  h_T(a)+h_T(b)\ge t\ge h_T(a+b),
  \qquad
  q_T(a)+q_T(b)\ge 1\ge q_T(a+b),
\]
and if one of \(a,b\) is zero, the inequalities are immediate.  Since
\(g_1,g_2\) are subadditive on the seed goods, the seed parts
\[
  S\mapsto g_i(x(S),y(S))
\]
are subadditive.  The maximum of two nonnegative subadditive functions is
subadditive, because
\[
  \max\{a+a',b+b'\}
  \le
  \max\{a,b\}+\max\{a',b'\}.
\]
Therefore \(v_1\) and \(v_2\) are subadditive.  The auxiliary valuations
\(v_k=q_T\circ z\) are subadditive by subadditivity of \(q_T\).  Hence all
valuations are monotone and subadditive.

We will use two simple facts about the seed instance.  These are just the
two-agent tables from \Cref{thm:subadditive-alpha} written in threshold form.

First, since \(t\) lies between \(\frac12+\varepsilon\) and \(r\), the cells
above \(t\) are described by
\[
  g_1(x,y)>t
  \quad\Longleftrightarrow\quad
  x=3\ \text{ or }\ y\ge4\ \text{ or }\ (x\ge2\text{ and }y\ge3),
\]
and
\[
  g_2(x,y)>t
  \quad\Longleftrightarrow\quad
  y=5\ \text{ or }\ (x\ge1\text{ and }y\ge2).
\]
Thus, if two seed count vectors \(c_1,c_2\) fit inside the seed goods, meaning
\[
  c_1+c_2\le(3,5)
\]
coordinatewise, and
\[
  g_1(c_1)>t,
  \qquad
  g_2(c_2)>t,
\]
then there are only two possibilities:
\[
  (c_1,c_2)=((3,0),(0,5))
  \quad\text{or}\quad
  (c_1,c_2)=((2,3),(1,2)).
\]
Indeed, if \(c_2\) has \(y_2=5\), then \(c_1\) has \(y_1=0\), so the first
threshold rule forces \(x_1=3\), giving \(((3,0),(0,5))\).  Otherwise the
second threshold rule gives \(x_2\ge1\) and \(y_2\ge2\).  Then
\(x_1\le2\) and \(y_1\le3\), so the first threshold rule forces
\(c_1=(2,3)\), and hence \(c_2=(1,2)\).

Both of these possibilities are incompatible with \EFone between agents
\(1\) and \(2\), even if auxiliary goods are also present.  For
\(((3,0),(0,5))\), agent \(1\)'s own value is
\[
  \max\{g_1(3,0),h_T(z_1)\}=r,
\]
since \(h_T\le t<r\).  After any one-good deletion from agent \(2\)'s bundle,
the remaining seed count is either \((0,5)\) or \((0,4)\), and agent \(1\)
values both at least \(r+\varepsilon\).  Hence agent \(1\) still envies agent
\(2\) after every one-good deletion.

For \(((2,3),(1,2))\), agent \(2\)'s own value is
\[
  \max\{g_2(1,2),h_T(z_2)\}=r.
\]
After any one-good deletion from agent \(1\)'s bundle, the remaining seed count
is one of
\[
  (2,3),\qquad (1,3),\qquad (2,2),
\]
and agent \(2\) values each of these at least \(r+\varepsilon\).  Hence agent
\(2\) still envies agent \(1\) after every one-good deletion.

Second, the top seed cells of one agent are surrounded, in the other agent's
table, by cells of value at least \(r\).  More precisely:
\[
  g_2(c)>r+2\varepsilon
  \quad\Longrightarrow\quad
  g_1(c)\ge r
  \ \text{and every one-good seed deletion of }c\text{ has }g_1\text{-value at least }r,
\]
and symmetrically,
\[
  g_1(c)>r+2\varepsilon
  \quad\Longrightarrow\quad
  g_2(c)\ge r
  \ \text{and every one-good seed deletion of }c\text{ has }g_2\text{-value at least }r.
\]
Here is the short verification.  The condition \(g_2(c)>r+2\varepsilon\)
means that \(g_2(c)=1\).  In the \(g_2\)-table, this occurs only when
\[
  y=5,
  \qquad
  \text{or } x\ge2,\ y\ge4,
  \qquad
  \text{or } x=3,\ y\ge3.
\]
In each case, \(c\) and every available predecessor satisfy the threshold rule
for \(g_1\ge r\) above.  The symmetric statement is the same: the cells with
\(g_1(c)=1\) are
\[
  x=1,\ y=5,
  \qquad
  \text{or } x\ge2,\ y\ge3,
\]
and \(c\) and every available predecessor satisfy the threshold rule for
\(g_2\ge r\).

Now let \(X\) be any \EFone allocation.  We will show that \(X\) is not
\(\alpha\)-Pareto optimal.

For agents \(1\) and \(2\), write
\[
  c_i=(x(X_i),y(X_i)),
  \qquad
  z_i=z(X_i).
\]
For every auxiliary agent \(k\ge3\), write \(z_k=z(X_k)\).

We first prove that no auxiliary agent can receive \(T\) or more auxiliary
goods in an \EFone allocation.  Suppose, for contradiction, that some auxiliary
agent \(k\) has
\[
  z_k\ge T.
\]
After deleting any one good from \(X_k\), at least \(T-1\) auxiliary goods
remain.  Hence every other auxiliary agent \(\ell\ge3\) must have value at
least
\[
  q_T(T-1)=r;
\]
otherwise \(\ell\) would still envy \(k\) after every one-good deletion.  Thus
every auxiliary agent other than \(k\) must receive at least \(T-1\) auxiliary
goods.

The auxiliary agents therefore hold at least
\[
  T+(m-1)(T-1)=mT-(m-1)
\]
auxiliary goods.  Since there are \(mT\) auxiliary goods in total, agents
\(1\) and \(2\) together hold at most \(m-1\) auxiliary goods.  Because
\(T\ge m+1\), this is at most \(T-2\).  Therefore
\[
  h_T(z_1)\le\frac t2,
  \qquad
  h_T(z_2)\le\frac t2.
\]

Now look at agent \(i\in\{1,2\}\).  After any one-good deletion from
\(X_k\), at least \(T-1\) auxiliary goods remain, so agent \(i\) values the
remaining bundle at least
\[
  h_T(T-1)=t.
\]
Since \(X\) is \EFone, agent \(i\)'s own value must be at least \(t\).  But
her auxiliary value is at most \(t/2\), so her seed value must be at least
\(t\).  No seed-table entry equals \(t\), and hence
\[
  g_i(c_i)>t
  \qquad\text{for both }i=1,2.
\]
By the first seed fact above, this forces one of the two seed configurations
that already fail \EFone between agents \(1\) and \(2\), a contradiction.
Therefore
\[
  z_k\le T-1
  \qquad\text{for every auxiliary agent }k\ge3.
\]
Consequently,
\[
  v_k(X_k)=q_T(z_k)\le r
  \qquad\text{for every }k\ge3.
\]

Next we bound the values of agents \(1\) and \(2\).  They cannot both have
value above \(t\).  Indeed, if
\[
  v_1(X_1)>t
  \qquad\text{and}\qquad
  v_2(X_2)>t,
\]
then, because \(h_T\le t\), both inequalities must come from the seed parts:
\[
  g_1(c_1)>t,
  \qquad
  g_2(c_2)>t.
\]
The first seed fact again says that \(X\) fails \EFone between agents \(1\)
and \(2\).  Hence at least one of agents \(1\) and \(2\) has value at most
\(t\).

Assume first that
\[
  v_1(X_1)\le t.
\]
We claim that then
\[
  v_2(X_2)\le r+2\varepsilon.
\]
If not, then \(v_2(X_2)>r+2\varepsilon\).  Since \(h_T\le t<r\), this must
come from the seed part:
\[
  g_2(c_2)>r+2\varepsilon.
\]
By the second seed fact, agent \(1\) values \(c_2\), and every one-good seed
deletion from \(c_2\), at least \(r\).  Therefore, after any one-good deletion
from \(X_2\), agent \(1\) still values the remaining bundle at least \(r\):
deleting an auxiliary good leaves the seed count \(c_2\) unchanged, and
deleting a seed good gives one of the seed deletions covered above.  Since
\[
  r>t\ge v_1(X_1),
\]
agent \(1\)'s envy of agent \(2\) cannot be removed by deleting one good.  This
contradicts \EFone.  Thus
\[
  v_1(X_1)\le t,
  \qquad
  v_2(X_2)\le r+2\varepsilon,
  \qquad
  v_k(X_k)\le r\quad(k\ge3).
\]

Now we use the same improving seed allocation as in the two-agent proof.  Give
all three \(A\)-goods to agent \(1\), all five \(B\)-goods to agent \(2\), and
give exactly \(T\) auxiliary goods to every auxiliary agent.  Call this
allocation \(Y\).  Then
\[
  v_1(Y_1)=g_1(3,0)=r,
  \qquad
  v_2(Y_2)=g_2(0,5)=1,
  \qquad
  v_k(Y_k)=q_T(T)=1 \quad(k\ge3).
\]
By the definition of \(\rho\),
\[
  r>\frac1\alpha\,t
    \ge \frac1\alpha\,v_1(X_1),
\]
\[
  1>\frac1\alpha\,(r+2\varepsilon)
    \ge \frac1\alpha\,v_2(X_2),
\]
and
\[
  1>\frac1\alpha\,r
    \ge \frac1\alpha\,v_k(X_k)
    \qquad(k\ge3).
\]
Thus \(Y\) improves every agent by a factor strictly larger than \(1/\alpha\).
So \(X\) is not \(\alpha\)-Pareto optimal.

The other case is symmetric.  If
\[
  v_2(X_2)\le t,
\]
then the same argument gives
\[
  v_1(X_1)\le r+2\varepsilon,
  \qquad
  v_k(X_k)\le r\quad(k\ge3).
\]
Now use the other improving seed allocation from the two-agent proof: give seed
counts \((2,3)\) to agent \(1\), seed counts \((1,2)\) to agent \(2\), and
give exactly \(T\) auxiliary goods to every auxiliary agent.  Call this
allocation \(Y'\).  Then
\[
  v_1(Y'_1)=g_1(2,3)=1,
  \qquad
  v_2(Y'_2)=g_2(1,2)=r,
  \qquad
  v_k(Y'_k)=q_T(T)=1 \quad(k\ge3).
\]
Again, by the definition of \(\rho\),
\[
  1>\frac1\alpha\,(r+2\varepsilon)
    \ge \frac1\alpha\,v_1(X_1),
\]
\[
  r>\frac1\alpha\,t
    \ge \frac1\alpha\,v_2(X_2),
\]
and
\[
  1>\frac1\alpha\,r
    \ge \frac1\alpha\,v_k(X_k)
    \qquad(k\ge3).
\]
So \(Y'\) also improves every agent by a factor strictly larger than
\(1/\alpha\).

Since every \EFone allocation \(X\) falls into one of the two cases above, every
\EFone allocation is \(\alpha\)-Pareto dominated.  Hence no allocation is both
\EFone and \(\alpha\)-Pareto optimal.

\end{document}